\documentclass[a4paper]{amsart}

\usepackage{amsmath}
\usepackage{amssymb}
\usepackage{amsthm}
\usepackage{graphicx}
\usepackage{bbm}
\usepackage{url}
\usepackage{color}
\usepackage{enumerate}
\usepackage{enumitem}

\newtheorem{theorem}{Theorem}[section]
\newtheorem{proposition}[theorem]{Proposition}
\newtheorem{lemma}[theorem]{Lemma}
\newtheorem{definition}[theorem]{Definition}

\newtheorem{algorithm}[theorem]{Algorithm}
\newtheorem{remark}[theorem]{Remark}

\numberwithin{equation}{section}

\title{Moment Explosions in the Rough Heston Model}

\author{Stefan Gerhold, Christoph Gerstenecker, Arpad Pinter}
\email{sgerhold@fam.tuwien.ac.at}
\address{FAM, TU Wien}
\thanks{Financial support from the
Austrian Science Fund (FWF) under grants P~24880 and P~30750
is gratefully acknowledged. We thank Omar El~Euch, Antoine Jacquier, and
Martin Keller-Ressel for helpful comments.}
\date{\today}

\begin{document}

\begin{abstract}
   We show that the moment explosion time in the
   rough Heston model [El Euch, Rosenbaum 2016, arxiv:1609.02108] is finite if and only if
   it is finite for the classical Heston model. Upper and lower bounds for the
   explosion time are established, as well as an algorithm to compute the
   explosion time (under some restrictions). We show that the critical moments are finite for
   all maturities. For negative correlation, we apply our algorithm for the moment explosion
   time to compute the lower critical moment.
\end{abstract}

%\keywords{}


\maketitle

\section{Introduction}

It has long been known that the marginal distributions of a realistic asset price model
should not feature tails that are too thin (as, e.g., in the Black-Scholes model). In many models that have been proposed, the tails are of power law type. Consequently, not all moments of the asset price are finite.
Existence of the moments has been thoroughly investigated for classical models; in particular,
we mention here Keller-Ressel's work~\cite{KR11} on affine stochastic volatility models. Precise information
on the critical moments -- the exponents where the stock price ceases to be integrable,
depending on maturity -- is of interest for several reasons.
It allows to approximate the wing behavior of the volatility smile, to assess the convergence
rate of some numerical procedures, and to identify models that would assign infinite
prices to certain financial products.
We refer to~\cite{AnPi07,KR11} and the article
\emph{Moment Explosions} in~\cite{Co14}
for further details and references on these motivations.
Moreover, when using the Fourier representation to price options, choosing
a good integration path (equivalently, a good damping parameter) to avoid
highly oscillatory integrands requires knowing
the strip of analyticity of the characteristic function. Its boundaries are described
by the critical moments~\cite{Le04b,LoKa07}.

In recent years, attention has shifted in financial modeling from the classical (jump-)diffusion and L\'evy models
to rough volatility models. Since the pioneering work by Gatheral et al.~\cite{GaJaRo14},
the literature on these non-Markovian stochastic volatility models, inspired by fractional
Brownian motion, has grown rapidly. We refer, e.g., to Bayer et al.~\cite{BaFrGaMaSt17} for many
references.
%For rough volatility models, not much is known concerning moment explosions.
In the present paper we provide some results
on the explosion time and the critical moments of the rough Heston model.
While there are several ``rough'' variants of the Heston model, we work with
the one proposed by El Euch and Rosenbaum~\cite{ElRo16}.
The dynamics of this model are given by
\begin{align*}
  dS_t &= S_t\sqrt{V_t}\,dW_t,\quad S_0 > 0, \\ 
  V_t &= V_0 + \frac{1}{\Gamma(\alpha)}\int_0^t (t-s)^{\alpha-1}\lambda(\bar{v}-V_s)\,ds + \frac{1}{\Gamma(\alpha)}\int_0^t (t-s)^{\alpha-1}\xi\sqrt{V_s}\,dZ_s,  \\
  d\langle W,Z\rangle_t &= \rho\,dt,
\end{align*}
where $W$ and $Z$ are correlated Brownian motions, $\rho\in(-1,1)$, and $\lambda,$ $\xi,$ $\bar{v}$
are positive parameters. The smoothness parameter $\alpha$ is in $(\tfrac12,1)$.
(For  $\alpha=1$, the model clearly  reduces to the classical Heston model.) Besides having a
microstructural foundation, this model features a characteristic function that can
be evaluated numerically in an efficient way, by solving a fractional Riccati equation
(equivalently, a non-linear Volterra integral equation; see Section~\ref{se:prelim}).
Its tractability
makes the rough Heston model attractive for practical implementation, and at the same time facilitates our analysis.

We first analyze the explosion time, i.e., the maturity at which a fixed moment
explodes. While the explosion time of the classical Heston model has an explicit formula,
for rough Heston we arrive at a well-known hard problem: Computing the explosion
time of the solution of a non-linear Volterra integral equation (VIE) of the second kind.
There is no general algorithm known, and in most cases that have been studied
in the literature, only bounds are available.
See Roberts~\cite{Ro98} for an overview.
Using the specific structure of our case, we show that the explosion time is finite
if and only if it is finite for the  classical Heston model,  and we provide a lower and an
upper bound (Sections~\ref{se:expl vie}--\ref{se:val}).
As a byproduct, the validity of the fractional Riccati equation, respectively the VIE,
for all moments is established, which culminates in Section~\ref{se:complex}.
In Section~\ref{se:comp exp} we derive an algorithm to compute the explosion time, under some restrictions
on the parameters.
%\margincomment{mention monotonicity and Volterra, when done}
The critical moments are finite for all maturities (Section~\ref{se:cm fin}) and
can be computed by numerical root finding (Section~\ref{se:comp mom}).
Our approach has two limitations: First, to compute the critical moments, maturity must
not be too high. Second, our algorithm can compute the upper critical moment only for $\rho>0$,
and the lower critical moment for $\rho<0$. As the latter is the more important case in
practice, we focus on the \emph{left} wing of implied volatility when recalling
the relation between critical moments and strike asymptotics (Lee's moment formula; see Section~\ref{se:app}).

Corollary~3.1 in~\cite{ElRo17} is related to our results. For each maturity, it gives explicit lower and
upper bounds for the critical moments. Inverting them yields a lower bound
for the explosion time; the latter is not comparable to our bounds.

\section{Preliminaries}\label{se:prelim}

El Euch and Rosenbaum~\cite{ElRo16} established a semi-explicit representation of the moment generating function (mgf) of the log-price $X_t = \log(S_t/S_0)$
in the rough Heston model.
The mgf is given by
\begin{equation}\label{eq:mgf}
  \mathbb{E}[e^{uX_t}] = \exp\big(\bar{v}\lambda I^1_t \psi(u,t)+v_0I^{1-\alpha}_t \psi(u,t)\big),
\end{equation}
where $\psi$ satisfies a fractional Riccati differential equation (see below).
The constraint $\rho\in(-1/\sqrt{2},1/\sqrt{2}]$ from~\cite{ElRo16} was removed recently
in~\cite{ElRo17}.
The paper~\cite{AlGaRa17} contains an alternative derivation of the fractional
Riccati equation, and~\cite{AbLaPu17}
has more general results, embedding the rough Heston model into the new class
of affine Volterra processes.
Recall the following definition (see e.g.~\cite{KiSrTr06}):
\begin{definition}%[Riemann-Liouville fractional integral and derivative]
  The (left-sided) Riemann-Liouville fractional integral $I_t^\alpha$ of order $\alpha\in(0,\infty)$ %started at $0$ with respect to the integration variable~$t$
  of a function~$f$ is given by
  \begin{equation}\label{eqIII:frac-int}
    I^\alpha_t f(t) := \frac1{\Gamma(\alpha)}\int_0^t (t-s)^{\alpha-1}f(s)\,ds
  \end{equation}
  whenever the integral exists, and the (left-sided) Riemann-Liouville fractional derivative $D_t^\alpha$ of order $\alpha\in[0,1)$ %started at $0$ with respect to the integration variable~$t$
  of~$f$ is given by
  \begin{equation}\label{eqIII:frac-der}
    D^\alpha_t f(t) := \frac{d}{dt}I^{1-\alpha}f(t)
  \end{equation} 
  whenever this expression exists. 
\end{definition}
(The fractional derivative $D^\alpha_t$ can be defined for $\alpha>1$ as well, but
this is not needed in our context.)
The function~$\psi$ from~\eqref{eq:mgf}  is the unique continuous solution of the fractional Riccati initial value problem
\begin{align}
  D^\alpha_t \psi(u,t) &= R(u,\psi(u,t)), \label{eqIII:Ha1}\\
  I^{1-\alpha}_t \psi(u,0) &= 0, \label{eqIII:Ha2}
\end{align}
where~$R$ is defined as
\begin{equation}\label{eqIII:R}
  R(u,w)=c_1(u)+c_2(u)w+ c_3w^2,
\end{equation}
with coefficients
\begin{align*}
  c_1(u)&=\tfrac12 u(u-1), \\ 
  c_2(u)&=\rho\xi u-\lambda, \\
  c_3&=\tfrac12\xi^2>0.
\end{align*}
For $\alpha=1$, this becomes a standard Riccati differential equation, which admits
a well-known explicit solution~\cite[Chapter~2]{Ga06}.
The roots of $R(u,\cdot)$ are located at the points $\frac1{c_3}(-e_0(u)\pm\sqrt{e_1(u)})$
with 
\begin{align}
  e_0(u) &:= \tfrac12 c_2(u) = \tfrac12(\rho\xi u-\lambda), \label{eqIII:e_0} \\
  e_1(u) &:= e_0(u)^2-c_3c_1(u) = e_0(u)^2 - \tfrac14\xi^2 u(u-1).\label{eqIII:e_1}
\end{align}

The following result, relating fractional differential equations
and Volterra integral equations, is a special case of Theorem~3.10 in~\cite{KiSrTr06}.
\begin{theorem}\label{thm:frac vie}
  Let $\alpha\in(0,1)$, $T>0$ and suppose that $\psi \in C[0,T]$ and $H\in C(\mathbb R)$.
  Then~$\psi$ satisfies the fractional differential equation
  \begin{align*}
    D_t^\alpha \psi(t) &= H(\psi(t)), \\
    I_t^{1-\alpha}\psi(0) &= 0
  \end{align*}
  if and only if it satisfies the Volterra integral equation (VIE)
  \[
    \psi(t) = \frac{1}{\Gamma(\alpha)}\int_{0}^t(t-s)^{\alpha-1} H(\psi(s))ds.
  \]
\end{theorem}
Using Theorem~\ref{thm:frac vie},
the Riccati differential equation~\eqref{eqIII:Ha1} with initial value~\eqref{eqIII:Ha2} can be transformed into the non-linear Volterra integral equation
\begin{equation*}
  \psi(u,t) = \frac1{\Gamma(\alpha)}\int_0^t (t-s)^{\alpha-1} R(u,\psi(u,s))\,ds.
\end{equation*}
%with the weakly singular kernel $k(z):=\frac1{\Gamma(\alpha)}z^{\alpha-1}$.
This integral equation was used in~\cite{ElRo16} to compute~$\psi$ numerically.
The function
\[
  f(u,t):=c_3\psi(u,t)
\]
solves
\begin{equation}\label{eq:vie}
  f(u,t) = \frac1{\Gamma(\alpha)}\int_0^t (t-s)^{\alpha-1} G(u,f(u,s))\,ds
\end{equation}
with non-linearity
\begin{align}
  G(u,w) &= c_3R(u,w/c_3) \notag \\
  &= (w + e_0(u))^2- e_1(u),  \label{eqIII:G1}
\end{align}
where $e_0(u)$ and $e_1(u)$ are defined in~\eqref{eqIII:e_0} and~\eqref{eqIII:e_1}.
Equation~\eqref{eq:vie} is a nonlinear Volterra integral equation with weakly singular kernel;
it will be used to analyze the blow-up behavior of~$f$ (and thus of~$\psi$) in Section~\ref{se:expl vie}.
We quote the following standard existence and uniqueness result for equations of this kind.
\begin{theorem}\label{thm:vie ex}
  Let $\alpha\in(0,1)$, $g\in C[0,\infty)$, and suppose that $H:\mathbb R \to \mathbb R$
  is locally Lipschitz continuous. Then there is $T^*\in(0,\infty]$ such that
  the Volterra integral equation
  \begin{equation*}%\label{eq:gen vie}
    \psi(t) = g(t) + \int_0^t (t-s)^{\alpha-1} H(u(s))ds
  \end{equation*}
  has a unique continuous solution $\psi$ on $[0,T^*)$, 
  and there is no continuous solution on any larger right-open interval $[0,T^{**})$.
\end{theorem}
\begin{proof}
For existence and uniqueness on a sufficiently small interval $[0,T_0]$ with $T_0>0$ see,
e.g., Theorem~3.1.4
in Brunner's recent monograph~\cite{Br17}.
The continuation to a maximal right-open interval is discussed there as well (p.~107;
see also Section~12 of Gripenberg et al.~\cite{GrLoSt90}).
\end{proof}
In~\cite{ElRo16}, the fractional Riccati equation was established
 for $u\in i\mathbb R$, whereas we are interested
in $u\in\mathbb R$. The justification of~\eqref{eqIII:Ha1}--\eqref{eqIII:Ha2},
and thus of~\eqref{eq:vie}, for $u\in\mathbb R$ hinges on a result from~\cite{ElRo17}
and the analytic dependence of $f(u,t)$ on~$u$. See Sections~\ref{se:val}
and~\ref{se:complex} for details.
%We write $T^\ast_{\alpha}(u) \in (0,\infty]$ for the explosion time of the solution
%of~\eqref{eq:vie}, i.e., the number~$T^*$ from Theorem~\ref{thm:vie ex}.
%By Theorem~\ref{thm:cons}, this notation is consistent with~\eqref{eq:def exp time} below.
%
We write
\begin{equation}\label{eq:def exp time}
  T^\ast_{\alpha}(u):=\sup\{t\geq 0\colon \mathbb{E}[S_t^u]<\infty\}, \quad u\in\mathbb R,
\end{equation}
for the moment explosion time in the rough Heston model.
In the classical case ($\alpha=1$), the explicit expression
\begin{align}
  T^\ast_{1}(u) &=
    \begin{cases}
    \int_0^\infty \frac1{R(u,w)}\,dw, & \text{$R(u,\cdot)$ has no roots on $[0,\infty)$,} \\
    \infty, & \text{otherwise,}
  \end{cases} \label{eq:He et1} \\ 
  &=\begin{cases}
    \frac1{\sqrt{|e_1(u)|}}\left(\frac\pi2-\arctan\left(\frac{e_0(u)}{\sqrt{|e_1(u)|}}\right)\right), & e_1(u) < 0, \\
    \frac1{2\sqrt{e_1(u)}}\log\left(\frac{e_0(u)+\sqrt{e_1(u)}}{e_0(u)-\sqrt{e_1(u)}}\right), & e_1(u)\geq 0,e_0(u)>0, \\
    \infty, & e_1(u)\geq 0, e_0(u)<0, \label{eq:He et}
  \end{cases}
\end{align}
has been found by Andersen and Piterbarg~\cite{AnPi07} (see also~\cite{KR11}).
It is a consequence of the explicit characteristic function, which is not available
for the \emph{rough} Heston model.
We distinguish between the following cases for $u\in\mathbb{R}$:
\begin{enumerate}[label=(\Alph*)]
\item\label{itemIII: A} $c_1(u)>0$, $e_0(u)\geq 0,$
\item\label{itemIII: B} $c_1(u)>0$, $e_0(u)<0$ and $e_1(u)<0,$
\item\label{itemIII: C} $c_1(u)>0$, $e_0(u)<0$ and $e_1(u)\geq 0,$
\item\label{itemIII: D} $c_1(u)\leq 0.$
\end{enumerate}
As several of our results deal with $\rho<0$ and $u$ satisfying case~(A), we explicitly note
\[
  \text{case (A) for $\rho<0$:}\qquad u\leq \frac{\lambda}{\rho\xi}.
\]
Note that cases~\ref{itemIII: A} and~\ref{itemIII: B} combined are exactly the cases in which the moment explosion time $T^\ast_{1}(u)$ in the classical Heston model is finite, by~\eqref{eq:He et}.
We can now state our first main result.
\begin{theorem}\label{thm:finite}
  For $u\in\mathbb R$, the moment explosion time $T_\alpha^\ast(u)$ of the rough
  Heston model is finite if and only if~$u$ satisfies (A) or (B). This is equivalent
  to $T_1^\ast(u)$ (explosion time of the classical  Heston model) being finite.
\end{theorem}
The proof of Theorem~\ref{thm:finite} consists of two main parts.
First, Propositions~\ref{prop:case1}, \ref{prop:case2},
\ref{prop:case3}, and~\ref{prop:case4} discuss the blow-up behavior of the solution of~\eqref{eq:vie}
in cases (A)-(D), and Lemma~\ref{le:blowup} shows that blow-up of~$f$ leads (unsurprisingly)
to blow-up of the right-hand side of~\eqref{eq:mgf}. Second, we show in Section~\ref{se:val} that
the explosion time of~$f(u,\cdot)$ (the solution of~\eqref{eq:vie}) agrees with~$T_\alpha^*(u)$
(the explosion time of the rough Heston model) for all $u\in\mathbb R$.
As mentioned after Theorem~\ref{thm:vie ex},
this is not obvious from the results in the existing literature.

\section{Explosion time of the Volterra integral equation}\label{se:expl vie}

We begin by citing a result from Brunner and Yang~\cite{BrYa12} which characterizes the blow-up behavior of non-linear Volterra integral equations defined by positive and increasing functions.
We note that some arguments in our subsequent proofs (from Proposition~\ref{prop:case1} onwards)
are similar to arguments used in~\cite{BrYa12}. Alternatively, it should be possible
to extend the arguments in Appendix~A of~\cite{GaKe18}; there, $u$ is in~$[0,1]$.

%-----------------------------------------------------------------------------------------------------

\begin{proposition}\label{propIII:Brunner}
  Assume that $G\colon[0,\infty)\to[0,\infty)$ is continuous and the following conditions hold:
  \begin{itemize}
  \item[(G1)] $G(0)=0$ and $G$ is strictly increasing,
  \item[(G2)] $\lim_{w\to\infty} G(w)/w = \infty$,
  \item[(P)] $\phi\colon[0,\infty)\to(0,\infty)$ is a positive, non-decreasing, continuous function,
  \item[(K)] $k\colon(0,\infty)\to[0,\infty)$ is locally integrable and $K(t):=\int_0^t k(z)\,dz>0$ is a non-decreasing function.
  \end{itemize}
  Furthermore, assume $\lim_{t\to\infty}\phi(t)=\infty$ and $k(z)=cz^{\alpha-1}$ with $\alpha>0$ and $c>0$. Then the solution $h$ of the Volterra integral equation
  \[
    h(t) = \phi(t) + \int_0^t k(t-s)G(h(s))\,ds
  \]
  blows up in finite time if and only if
  \begin{equation}\label{eqIII:blowup-int}
    \int_U^\infty \left(\frac{w}{G(w)}\right)^{1/\alpha}\frac{dw}{w}<\infty
  \end{equation}
  for all $U>0$.
\end{proposition}

%-----------------------------------------------------------------------------------------------------

\begin{proof}
  This is a special case of Corollary 2.22 in Brunner and Yang~\cite{BrYa12}, with $G$ not depending on time.
\end{proof}

%-----------------------------------------------------------------------------------------------------

In case~\ref{itemIII: A}, all assumptions of Proposition~\ref{propIII:Brunner} are satisfied and only the integrability condition~\eqref{eqIII:blowup-int} has to be checked to determine whether the solution $f$ of \eqref{eq:vie} blows up in finite time.

%-----------------------------------------------------------------------------------------------------

\begin{proposition}\label{prop:case1}
  In case~\ref{itemIII: A}, the solution $f$ of~\eqref{eq:vie} starts at $0$, is positive thereafter and blows up in finite time.
\end{proposition}

%-----------------------------------------------------------------------------------------------------

\begin{proof}
  Fix $u\in\mathbb{R}$ such that $c_1(u)>0$ and $e_0(u)\geq 0$. Note that $e_0^2-e_1>0$ in this case.
   (Here and in the following, we will often suppress~$u$ in the notation.)
   If we write the Volterra integral equation~\eqref{eq:vie} in the form
  \[
    f(t) = \phi(t) + \int_0^t k(t-s) \bar G(f(s))\,ds
  \]
  with non-linearity $\bar G(w) = w^2 + 2e_0 w$ and $\phi(t) = \frac{e_0^2-e_1}{\Gamma(1+\alpha)} t^\alpha$, using~\eqref{eqIII:G1} and~\eqref{eqIII:frac1a}, then the conditions $c_1>0$ and $e_0\geq 0$ guarantee that $\phi$ and $\bar G$ are positive and strictly increasing on $(0,\infty)$ with $\phi(0)=\bar G(0)=0$. Hence, the solution~$f$ is positive for positive values, and $f(0)=0$. 
  (Positivity follows from Lemma~2.4 in~\cite{BrYa12}, or from Lemma~3.2.11 in~\cite{Br17}.)
  It is easy to check that all the assumptions (G1), (G2), (P) and (K) of Proposition~\ref{propIII:Brunner} are satisfied. Moreover, $\lim_{t\to\infty}\phi(t)=\infty$ and
  \begin{align*}
    \int^\infty_U \left(\frac{w}{\bar G(w)}\right)^{1/\alpha}\frac{dw}{w} \leq\int^\infty_U w^{-1-1/\alpha}\,dw <\infty
  \end{align*}
  for all $U>0$. By Proposition~\ref{propIII:Brunner}, the solution~$f$ blows up in finite time.
\end{proof}

%-----------------------------------------------------------------------------------------------------

In case~\ref{itemIII: B}, Proposition~\ref{propIII:Brunner} cannot be applied directly to the solution~$f$ of~\eqref{eq:vie}. Hence, the Volterra integral equation has to be modified in order to satisfy the assumptions of Proposition~\ref{propIII:Brunner} in a way that $f$ is still a subsolution of the modified equation, i.e.\ $f$ satisfies~\eqref{eq:vie} with ``$\geq$'' instead of ``$=$''. First, we provide a comparison lemma for solutions and subsolutions.

%-----------------------------------------------------------------------------------------------------

\begin{lemma}\label{eqIII:inequality}
  Let $G\colon[0,\infty)\to(0,\infty)$ be a strictly increasing, continuous function and $T>0$. Suppose that $g$ is the unique continuous solution of the Volterra integral equation
  \[
    g(t) = \int_0^t k(t-s) G(g(s))\,ds, \quad t\in[0,T],
  \]
  where~$k$ satisfies condition~(K) from Proposition~\ref{propIII:Brunner}.
  If $f$ is a continuous subsolution,
  \[
    f(t) \geq \int_0^t k(t-s)G(f(s))\,ds, \quad t\in[0,T],
  \]
  then $f(t)\geq g(t)$ holds for all $t\in[0,T]$.
\end{lemma}

%-----------------------------------------------------------------------------------------------------

\begin{proof}
  %The idea of the proof is based on the proof of Lemma~2.4 in~\cite{BrYa12}.  
  For any $c\in(0,T)$ define $f_c(t):=f(t+c)$ for $t\in[0,T-c]$. From the positivity of $G$, it follows that $f_c(0)=f(c)>0$ and 
  \[
    f_c(t)\geq\int_0^t k(t-s)G(f_c(s))\,ds, \quad t\in[0,T-c].
  \]
  Since $g(0)=0$, it follows that $g(0)<f_c(0)$. We want to show $g<f_c$ on the whole interval $[0,T-c]$. Therefore, suppose that $t\in(0,T-c]$ exists such that $0\leq g(s)<f_c(s)$ for all $s\in[0,t)$ and $g(t)=f_c(t)$. However, because $G$ is strictly increasing, we have
  \[
    0 = f_c(t)-g(t) \geq \int_0^t k(t-s)\big(G(f_c(s))-G(g(s))\big)\,ds > 0,
  \]
  which is a contradiction. Hence, the inequality $g(t)<f_c(t) = f(t+c)$ holds for all $t\in[0,T-c]$. Since $c\in(0,T)$ was arbitrary, the result follows easily.
\end{proof}

%-----------------------------------------------------------------------------------------------------

\begin{proposition}\label{prop:case2}
  In case~\ref{itemIII: B}, the solution $f$ of \eqref{eq:vie} starts at $0$, is positive thereafter and blows up in finite time.
\end{proposition}

%-----------------------------------------------------------------------------------------------------

\begin{proof}
  Fix $u\in\mathbb{R}$ such that $c_1(u)>0$, $e_0(u)<0$ and $e_1(u)<0$. Note that in this case, the non-linearity $G$ is obviously positive by \eqref{eqIII:G1}. However, $G$ is strictly \emph{decreasing} on $[0,-e_1]$. To deal with this problem, let $0<a<-e_1$ and define the modified non-linearity $\bar G_a$ as
  \begin{equation}\label{eqIII:G_a}
    \bar G_a(w)=
    \begin{cases}
      w\frac{a+e_1}{e_0} + a, & 0\leq w <-e_0, \\
      G(w), & w \geq -e_0.
    \end{cases}
  \end{equation}
  Then $\bar G_a$ is a positive, strictly increasing, Lipschitz continuous function that starts at $a$ and $\bar G_a\leq G$. Let $\bar f$ be the unique continuous solution (recall Theorem~\ref{thm:vie ex})
  of the Volterra integral equation
  \begin{equation}\label{eq:vie bar}
    \bar f(t) = \int_0^t k(t-s) \bar G_a(\bar f(s))\,ds = \phi(t) + \int_0^t k(t-s) \bar G(\bar f(s))\,ds
  \end{equation}
  with $\bar G = \bar G_a - a$ and $\phi(t)=\frac{a}{\Gamma(1+\alpha)}t^\alpha$.
  Note that the second equality in~\eqref{eq:vie bar} follows from~\eqref{eqIII:frac1a}. Due to the positivity of $\phi$ and $\bar G$ on $(0,\infty)$, the solution $\bar f$ is positive on $(0,\infty)$ as well. The functions $\phi$, $\bar G$ and $k$ satisfy the assumptions (G1), (G2), (P) and (K) in Proposition~\ref{propIII:Brunner}. Furthermore, $\lim_{t\to\infty}\phi(t)=\infty$ and $\bar G$ satisfies \eqref{eqIII:blowup-int}. By Proposition~\ref{propIII:Brunner}, $\bar f$ blows up in finite time. Because $f$ satisfies \eqref{eq:vie} and $\bar G_a\leq G$, it follows that $f$ is a subsolution of the modified Volterra integral equation, i.e.,
  \[
    f(t) \geq \int_0^t k(t-s)\bar G_a(f(s))\,ds.
  \]
  Now, Lemma~\ref{eqIII:inequality} implies $f\geq \bar f$. Consequently, $f$ blows up as well.
\end{proof}

%-----------------------------------------------------------------------------------------------------

Cases~\ref{itemIII: C} and~\ref{itemIII: D} are the cases where the solution $f$ of \eqref{eq:vie} does not blow up in finite time. In fact, $f$ does not blow up at all, as we will see. The following lemma provides the key argument for both cases.  

%-----------------------------------------------------------------------------------------------------

\begin{lemma}\label{lemIII:Case3}
  Let $G\colon[0,\infty)\to[0,\infty)$ be a Lipschitz continuous function that is positive on $[0,a)$ and $G\equiv 0$ on $[a,\infty)$ for an $a>0$. Then the unique continuous solution $f$
  of the Volterra integral equation
  \[
    f(t) = \int_0^t k(t-s) G(f(s))\,ds
  \]
  satisfies $0\leq f(t)\leq a$ for all $t\geq 0$.
\end{lemma}

%-----------------------------------------------------------------------------------------------------

\begin{proof}
  The non-negativity of $G$ implies $f \geq 0$. Suppose $t>0$ exists such that $f(t)>a$. By the continuity of $f$, there exists $0<t_0<t$ that satisfies $f(t_0) = a$ and $f(s)>a$ for all $s\in(t_0,t)$. From $G\equiv 0$ on $[a,\infty)$, we have
  \[
    \int_{t_0}^t k(t-s) G(f(s))\,ds = 0.
  \]
  Since $G$ is non-negative and $k$ is decreasing,
  \begin{align*}
    0&< f(t)-f(t_0) \\
    &= \int_{t_0}^t k(t-s) G(f(s))\,ds + \int_0^{t_0}\big(k(t-s) - k(t_0-s)\big) G(f(s))\,ds \\
    &= \int_0^{t_0}\big(k(t-s) - k(t_0-s)\big) G(f(s))\,ds \leq 0,
  \end{align*}
  which is a contradiction. Therefore, $f$ satisfies $0\leq f(t)\leq a$ for all $t\geq 0$. 
\end{proof}

%-----------------------------------------------------------------------------------------------------

\begin{proposition}\label{prop:case3}
  In case~\ref{itemIII: C}, the solution $f$ of \eqref{eq:vie} is non-negative and bounded, and exists globally.
\end{proposition}

%-----------------------------------------------------------------------------------------------------

\begin{proof}
  Fix $u\in\mathbb{R}$ such that $c_1(u)>0$, $e_0(u)<0$ and $e_1(u)\geq 0$.
  Note that the inequality $0\leq e_1=e_0^2-c_1c_3<e_0^2$ implies $a:=-e_0-\sqrt{e_1}>0$. Moreover, from \eqref{eqIII:G1}, it follows that $a$ is the smallest positive root of $G$. Define the non-linearity~$\bar G$ as
  \[
    \bar G(w):=
    \begin{cases}
      G(w), & 0\leq w \leq a, \\
      0, & w > a.
    \end{cases}
  \]
  Then $\bar G$ is a non-negative, Lipschitz continuous function that starts at $e_0^2-e_1>0$. Therefore, Lemma~\ref{lemIII:Case3} yields that the unique continuous solution $\bar f$ of
  \[
    \bar f(t) = \int_0^t k(t-s)\bar G(\bar f(s))\,ds
  \]
  is bounded with $0\leq\bar f(t)\leq a$ for all $t\geq 0$. Since $\bar G = G$ on $[0,a]$, the function $\bar f$ solves the original Volterra integral equation
  \[
    \bar f(t) = \int_0^t k(t-s) G(\bar f(s))\,ds
  \]
  and from the uniqueness of the solution we obtain $f=\bar f$.
\end{proof}

%-----------------------------------------------------------------------------------------------------

\begin{proposition}\label{prop:case4}
  In case~\ref{itemIII: D}, the solution~$f$ of~\eqref{eq:vie} is non-positive and bounded, and exists globally.
\end{proposition}

%-----------------------------------------------------------------------------------------------------

\begin{proof}
  Fix $u\in\mathbb{R}$ such that $c_1(u)\leq 0$, which is equivalent to $u\in[0,1]$.
  Note that $e_1=e_0^2-c_1c_3>e_0^2>0$ implies $a:=\sqrt{e_1}-e_0>0$. Moreover, from \eqref{eqIII:G1}, it follows that~$a$ is the smallest positive root of~$G$. Define $f_-:=-f$, which satisfies
  \begin{equation}\label{eqIII:h-}
    f_-(t) = -\int_0^t k(t-s)G(-f_-(s))\,ds.
  \end{equation}
  If we define the non-linearity $\bar G$ as
  \[
    \bar G(w):=
    \begin{cases}
      -G(-w), & 0\leq w \leq a, \\
      0, & w > a,
    \end{cases}
  \]
  then $\bar G$ is a non-negative, Lipschitz continuous function that starts at $e_1-\nolinebreak e_0^2>\nolinebreak 0$. With Lemma~\ref{lemIII:Case3} we obtain that the unique continuous solution $\bar f$ of
  \[
    \bar f(t) = \int_0^t k(t-s)\bar G(\bar f(s))\,ds
  \]
  is bounded with $0\leq\bar f(t)\leq a$ for all $t\geq 0$. Furthermore, $\bar f$ solves \eqref{eqIII:h-} because $\bar G(w) = -G(-w)$ for all $w\in[0,a]$. The uniqueness of the solution yields
  \[
    \bar f = f_-= -f.
  \]
  Hence, the solution $f$ is bounded with $-a\leq f(t)\leq 0$.
\end{proof}

%-----------------------------------------------------------------------------------------------------

We have shown that (A) and (B) are exactly the cases in which the solution $f$ of the Volterra integral equation~\eqref{eq:vie}, and thus the solution~$\psi$ of the fractional Riccati differential equation~\eqref{eqIII:Ha1} with initial value~\eqref{eqIII:Ha2}, blows up in finite time. 
The following lemma shows that blow-up of~$\psi$ is equivalent to
blow-up of the right-hand side of~\eqref{eq:mgf}. 
% which completes the proof of Theorem~\ref{thm:finite}.

%-----------------------------------------------------------------------------------------------------

\begin{lemma}\label{le:blowup}
  If $f$ is a non-negative, continuous function that blows up in finite time with explosion time $\hat T$, then $I^\alpha_tf$ blows up in finite time as well, with the same explosion time $\hat T$.  
  If $f$ is a bounded continuous function, then $I^\alpha_t f$ does not blow up in finite time. 
\end{lemma}

%-----------------------------------------------------------------------------------------------------

\begin{proof}
  First, suppose that the non-negative, continuous function $f$ explodes at $\hat T$ and let $M>0$. Then we can find $\varepsilon\in(0,\hat T/2)$ such that $f(t)\geq M$ for all $t\in(\hat T-\varepsilon,\hat T)$. Hence, 
  \begin{align*}
    I^\alpha_t f(t) &\geq \int_{T^\ast-\varepsilon}^t k(t-s)f(s)\,ds\\
    &\geq M \frac{(\hat T-\varepsilon)^\alpha}{\Gamma(1+\alpha)}\geq M\frac{(\hat T/2)^\alpha}{\Gamma(1+\alpha)}
  \end{align*}
  for all $t\in(\hat T-\varepsilon,\hat T)$.  
  For the second assertion, suppose that $f$ is continuous and bounded with $M>0$. Then we have
  \[
    |I^\alpha_t f(t)|\leq \frac{M}{\Gamma(1+\alpha)} t^\alpha
  \] 
  for all $t\geq 0$.
\end{proof}

\section{Bounds for the explosion time}\label{se:bounds}

We now establish lower and upper bounds for $\hat{T}_\alpha(u)$, valid whenever it is finite (cases~(A) and~(B)). We denote by $\hat{T}_\alpha(u)$ the
explosion time of the solution $f(u,\cdot)$ of~\eqref{eq:vie}.
As we will see later, it agrees with $T^*_\alpha(u)$, and so
both bounds of this section hold for the explosion time of the rough
Heston model. We prove them first, because we will apply the lower bound
in the proof of $T^*_\alpha(u)=\hat{T}_\alpha(u)$.

%-----------------------------------------------------------------------------------------------------

\begin{theorem}\label{thm:lb}
  In cases~\ref{itemIII: A} and~\ref{itemIII: B}, the blow-up time $\hat{T}_\alpha(u)$ of the solution $f(u,\cdot)$ of \eqref{eq:vie} satisfies
  \begin{equation}\label{eqIII:T*-lower}
    \hat{T}_\alpha(u) \geq\Gamma(1+\alpha)^{1/\alpha}\max_{r>1}\frac{(r^\alpha-1)^{1/\alpha}}{r(r-1)}\int_{a(u)}^\infty \bigg(\frac{w}{G(u,w)}\bigg)^{1/\alpha}\frac{dw}{w},
  \end{equation}
  where $a(u)=0$ in case~\ref{itemIII: A} and $a(u)=-e_0(u)>0$ in case~\ref{itemIII: B}.
\end{theorem}

%-----------------------------------------------------------------------------------------------------

\begin{proof}
  %The idea of the following proof is based on the proof of Lemma~2.19 in Brunner and Yang~\cite{BrYa12}.  
  Fix $u$ satisfying the requirements of case (A) or (B). It follows from Propositions~\ref{prop:case1} and~\ref{prop:case2} that in either case the solution $f$ is non-negative, starts at $0$ and $\lim_{t\uparrow \hat{T}_\alpha}f(t)=\infty$. For any $n\in\mathrm{N}_0$ choose
  \[
    t_n:=\min\{t>0: f(t)=(c r^n)^\alpha+a\}
  \]
  with $r>1$ and $c>0$.  
  Using the inequality $k(t_n-s)<k(t_{n-1}-s)$ for $s\in(0,t_{n-1})$, the non-negativity of $G$ and that $G$ is strictly increasing on $[a,\infty)$, we have for $n\in\mathrm{N}$
  \begin{align*}
    f(t_n) &= \int_0^{t_{n-1}}k(t_n-s)G(f(s))\,ds + \int_{t_{n-1}}^{t_n}k(t_n-s)G(f(s))\,ds \\
    &\leq f(t_{n-1}) + G(f(t_n))\int_{t_{n-1}}^{t_n}k(t_n-s)\,ds \\
    &= f(t_{n-1}) + \frac1{\Gamma(1+\alpha)} G(f(t_n)) (t_n-t_{n-1})^\alpha.
  \end{align*}
  Thus, we obtain for $n\in\mathrm{N}$
  \begin{align*}
    t_n-t_{n-1} &\geq \Gamma(1+\alpha)^{1/\alpha}\left(\frac{f(t_n)-f(t_{n-1})}{G(f(t_n))}\right)^{1/\alpha} \\
    &= \Gamma(1+\alpha)^{1/\alpha}(r^\alpha-1)^{1/\alpha} \frac{cr^{n-1}}{G((c r^n)^\alpha+a)^{1/\alpha}} \\
    &= \Gamma(1+\alpha)^{1/\alpha}\frac{(r^\alpha-1)^{1/\alpha}}{r(r-1)} \cdot\frac{cr^{n+1}-cr^n}{G((c r^n)^\alpha+a)^{1/\alpha}} \\
    &\geq \Gamma(1+\alpha)^{1/\alpha}\frac{(r^\alpha-1)^{1/\alpha}}{r(r-1)} \int_{c r^n}^{c r^{n+1}}\left(\frac1{G(s^\alpha+a)}\right)^{1/\alpha}\,ds.
  \end{align*}
  Finally,
  \begin{align*}
    \hat{T}_\alpha &= t_0 + \sum_{n=1}^\infty (t_n - t_{n-1}) \\
    &\geq \Gamma(1+\alpha)^{1/\alpha}\frac{(r^\alpha-1)^{1/\alpha}}{r(r-1)} \int_{c r}^\infty\left(\frac1{G(s^\alpha+a)}\right)^{1/\alpha}\,ds.
  \end{align*}
  Maximization over $c>0$, then $r>1$, and the substitution $w=s^\alpha+a$ yield the inequality \eqref{eqIII:T*-lower}.
\end{proof}

%-----------------------------------------------------------------------------------------------------

For $\alpha\uparrow 1$, the right-hand side of~\eqref{eqIII:T*-lower} simplifies to
\begin{equation}\label{eq:sharp}
  \int_{a(u)}^\infty \frac1{G(u,x)}\,dx = \int_{a(u)/c_3}^\infty \frac1{R(u,x)}\,dx.
\end{equation}
In case~\ref{itemIII: A}, the lower bound~\eqref{eqIII:T*-lower} is sharp
in the limit $\alpha\uparrow 1$: We have $a(u)=0$ then, and therefore~\eqref{eq:sharp} is exactly the moment explosion time~\eqref{eq:He et1}  of the classical Heston model.

Another lower bound for $\hat{T}_\alpha$ can be obtained from Corollary~3.1
in~\cite{ElRo17}. Numerical examples show that it is not comparable to the bound
from our Theorem~\ref{thm:lb}.

%-----------------------------------------------------------------------------------------------------

\begin{theorem}\label{thmIII:T*-upper}
  In cases~\ref{itemIII: A} and~\ref{itemIII: B}, the blow-up time $\hat{T}_\alpha(u)$ of the solution $f(u,\cdot)$ of \eqref{eq:vie} satisfies
  \begin{equation}\label{eqIII:T*-upper}
    \hat{T}_\alpha(u)\leq 4\Gamma(1+\alpha)^{1/\alpha}\int_{0}^\infty \bigg(\frac{w}{\hat{G}(u,w)}\bigg)^{1/\alpha}\frac{dw}{w},
  \end{equation}
  where $\hat{G}=G$ in case~\ref{itemIII: A}, and $\hat{G}\equiv -e_1$ on $[0,-e_0)$ and $\hat{G}=G$ on $[-e_0,\infty)$ in case~\ref{itemIII: B}.
\end{theorem}

%-----------------------------------------------------------------------------------------------------

\begin{proof}
  %The idea of the following proof is based on the proof of Lemma~2.12 in Brunner and Yang~\cite{BrYa12}.  
  Fix $u$ satisfying the requirements of case (A) or (B). From Propositions~\ref{prop:case1} and~\ref{prop:case2}, in either case the solution $f$ is positive on $(0,\infty)$, starts at $0$ and $\lim_{t\uparrow \hat{T}_\alpha}f(t)=\infty$. For any $n\in\mathrm{N}_0$ choose
  \begin{equation}\label{eqIII:t0}
    t_n:=\max\{t<\hat{T}_\alpha: f(t)=(c r^n)^\alpha\}
  \end{equation}
  with $r>1$ and $c>0$.  
  Define $\bar G:= G$ in case~\ref{itemIII: A} and $\bar G:=\bar G_a$ from \eqref{eqIII:G_a} in case~\ref{itemIII: B}. Since $\bar G\leq G$ and $\bar G$ is positive and strictly increasing, we have for $n\in\mathrm{N}$
  \begin{align*}
    f(t_n) &\geq \int_0^{t_n}k(t_n-s)\bar G(f(s))\,ds \\
    &\geq \int_{t_{n-1}}^{t_n}k(t_n-s)\bar G(f(s))\,ds \\
    &\geq \frac1{\Gamma(1+\alpha)}\bar G(f(t_{n-1})) (t_n-t_{n-1})^\alpha.
  \end{align*}
  Thus, we obtain for $n\in\mathrm{N}$
  \begin{align*}
    t_n-t_{n-1} &\leq \Gamma(1+\alpha)^{1/\alpha}\left(\frac{f(t_n)}{\bar G(f(t_{n-1}))}\right)^{1/\alpha} \\
    &= \Gamma(1+\alpha)^{1/\alpha}\frac{cr^n}{\bar G((c r^{n-1})^\alpha)^{1/\alpha}} \\
    &= \Gamma(1+\alpha)^{1/\alpha}\frac{r^2}{r-1}\cdot \frac{cr^{n-1}-cr^{n-2}}{\bar G((c r^{n-1})^\alpha)^{1/\alpha}} \\
    &\leq \Gamma(1+\alpha)^{1/\alpha}\frac{r^2}{r-1}\int_{c r^{n-2}}^{c r^{n-1}}\left(\frac1{\bar G(s^\alpha)}\right)^{1/\alpha}\,ds.
  \end{align*}
  Therefore,
  \begin{align*}
    \hat{T}_\alpha &= t_0 + \sum_{n=1}^\infty (t_n - t_{n-1}) \\
    &\leq t_0 + \Gamma(1+\alpha)^{1/\alpha}\frac{r^2}{r-1} \int_{c r^{-1}}^\infty\left(\frac1{\bar G(s^\alpha)}\right)^{1/\alpha}\,ds.
  \end{align*}
  Note that from the definition of $t_0$, it depends on $c>0$ and $r>1$. The fact that~$f$ is only zero at $t=0$ implies that $t_0\to 0$ as $c\downarrow 0$. Taking the limit $c\downarrow 0$, then minimizing over $r>1$ and substitution $w=s^\alpha$ yields
  \[
    \hat{T}_\alpha \leq 4\Gamma(1+\alpha)^{1/\alpha} \int_0^\infty\left(\frac{w}{\bar G(w)}\right)^{1/\alpha}\frac{dw}{w} 
  \]
  In case~\ref{itemIII: A}, we are finished. In case~\ref{itemIII: B}, we have $\bar G = \bar G_a$. Then the dominated convergence theorem for $a\uparrow -e_1$ yields the inequality \eqref{eqIII:T*-upper}.
\end{proof}
See Figures~\ref{fig:06}--\ref{fig:09} for numerical examples of these bounds.

\section{Explosion time in the rough Heston model}\label{se:val}

In Section~\ref{se:expl vie}, we established that the right-hand side of~\eqref{eq:mgf},
defined using the solution~$f$ of the VIE~\eqref{eq:vie},
explodes if and only if~$u$ satisfies the conditions of cases~(A) or~(B).
As before, we write $\hat{T}_\alpha(u)$ for the explosion time of $f(u,\cdot)$.
Recall that $T^*_\alpha(u)$ denotes the explosion time of the rough
Heston model, as defined in~\eqref{eq:def exp time}.
The goal of the present section is to show that $\hat{T}_\alpha(u)=T^*_\alpha(u)$,
and that~\eqref{eq:mgf} holds for all $u\in\mathbb R$ and $0<t<T^*_\alpha(u)$.
The following result from~\cite{ElRo17} was already mentioned at the
end of the introduction.
\begin{lemma}[Corollary~3.1 in~\cite{ElRo17}]\label{le:valER}
  For each $t>0$, there is an open interval such that~\eqref{eq:mgf} holds
  for all~$u$ from that interval.
\end{lemma}

\begin{lemma}
	The solution $ f $ of the Volterra integral equation \eqref{eq:vie} is differentiable w.r.t.\ $ u $, and its
	derivative satisfies
	\begin{align}\label{eq:LVIE for df/du}
	\partial_1 f(u, t)  = \int_{0}^{t} \frac{(t - s)^{\alpha - 1}}{\Gamma(\alpha)} \Big(  \partial_1 G(u, f(u, s))   +  \partial_2 G(u, f(u, s)) \partial_1 f(u, s)  \Big) ds.
	\end{align}
\end{lemma}
\begin{proof}
	We check the requirements of Theorem 13.1.2 in Gripenberg et al.~\cite{GrLoSt90}.
	%The first term vanishes in our equation.
	The polynomial $G(u,w) $ is differentiable. The kernel $(t - s)^{\alpha -1} / \Gamma(\alpha) =: k(t - s) $ is of \emph{continuous type} in the sense of~\cite{GrLoSt90}; 
	see the remark to Theorem 12.1.1 there, which states local integrability of $ k $ as a sufficient condition for this property.
	%The differentiability of $ a(u, t, s) $ in a speacial sense can be reached via $ a_u (u, t, s) \equiv 0 $, which is of continuous type too and now all the requirements of the Theorem in \cite{GrLoSt90} are fulfilled and we exactly get our claim.
\end{proof}

\begin{lemma}\label{le:mon}
  \begin{itemize}
    \item[(i)] In case (A) we have $ \partial_1 f(u, t) < 0 $ for $ u < 0 $ and $ \partial_1 f(u, t) > 0 $ for $ u > 0 $.
    \item[(ii)] If~$u$ satisfies case (B), then the same holds if $\hat{T}_\alpha(u)-t$
    is sufficiently small.
  \end{itemize}
\end{lemma}
\begin{proof}
    We only discuss the case $ u < 0 $, because $ u > 0 $ is analogous.
    
	(i) Note that~\eqref{eq:LVIE for df/du} is a ``\emph{linear} VIE'' that can be written as
	\begin{equation}\label{eq:LVIE}
		\partial_1 f(u, t) = g(t) + \int_{0}^{t} (t - s)^{\alpha - 1} K^{(u)}(t, s) \partial_1 f(u, s)\, ds,
	\end{equation}
	where we define
	\begin{align}
		g(t) &:= \int_{0}^{t} \frac{(t - s)^{\alpha - 1}}{\Gamma(\alpha)}
		  \partial_1 G(u, f(u, s))\, ds, \label{eq:def g}\\
		K^{(u)}(t, s) &= K^{(u)}(s):=\frac{\partial_2 G(u, f(u, s))}{\Gamma(\alpha)} \label{eq:def K}
	\end{align}
	to bring the notation close to that of Section~6.1.2 in~\cite{Br04}.
	Clearly,~\eqref{eq:LVIE} is not really a linear VIE, because the unknown function~$f$
	appears in~$g$ and $K^{(u)}$. But as our aim is not to solve it, but to control
	the sign of $\partial_1 f$, this viewpoint is good enough.
	
	As we are in case (A), we get from $ c_1(u) > 0 $ and $ e_0(u) \geq 0 $ that $ u \leq \lambda / (\rho \xi ) < 0 $. Furthermore, we have $ e_0'(u) = \xi \rho / 2 < 0 $, $ c_1'(u) = u - 1/2 < 0 $ and therefore
	\[
	f(u, s) e_0'(u) + \frac{1}{2} c_3 c_1'(u) < 0,
	\]
	since $ c_3 = \xi^2 / 2 > 0 $ and $ f(u, s) \geq 0 $ by Proposition~\ref{prop:case1}. From this we obtain $ \partial_1 G(u, f(u, s)) < 0 $, and hence  $ g(t) < 0 $ for all $ t \in [0, T_\alpha^*(u)) $.
	%The function $K^{(u)}(\cdot,\cdot)$ is continuous. 
	By Theorem 6.1.2 of Brunner~\cite{Br04}, we can
	express the solution of~\eqref{eq:LVIE for df/du} with the resolvent kernel
	$ R_{\alpha}(\cdot,\cdot)$,
	\begin{equation}\label{eq:vie der}
		\partial_1 f(u, t) = g(t) + \int_{0}^{t} R_{\alpha}(t, s) g(s)\, ds.
	\end{equation}
	The resolvent kernel has the explicit representation (see~\cite{Br04})
	\begin{equation}\label{eq:res kernel}
	  R_\alpha(t,s)= (t-s)^{\alpha-1}\sum_{n=1}^\infty Q_{\alpha,n}(t,s),
	\end{equation}
	where
	\begin{align}
	  Q_{\alpha,1}(t,s) &:= K^{(u)}(t, s) = K^{(u)}(s), \notag \\
	  Q_{\alpha,n}(t,s) &:= (t-s)^{(n-1)\alpha} \Phi_{\alpha,n}(t,s), \quad n\geq 2, \notag \\
	  \Phi_{\alpha,n}(t,s) &:= K^{(u)}(s) \int_0^1 (1-z)^{\alpha-1} z^{(n-1)\alpha-1}
	     Q_{\alpha,n-1}(t,s+(t-s)z)dz, \quad n\geq2. \label{eq:def Phi}
	\end{align}
	From this representation of the resolvent kernel, and the fact that~\eqref{eq:def K}
	is non-negative in case~(A), it is obvious
	that $ R_{\alpha} \geq 0 $. Since $g<0,$ we thus conclude from~\eqref{eq:vie der} that $ \partial_1 f(u, t) < 0 $ for all $ t \in [0, \hat{T}_\alpha(u)) $.
	
	(ii) Recall that we assume that $u<0$, because $u>0$ is analogous. We have to show that
	\begin{equation}\label{eq:tau}
	  \tau(u) := \inf \{ 0<t<\hat{T}_\alpha(u) : \partial_1 f(u,\cdot)<0\
	  \text{on}\ (t,\hat{T}_\alpha(u)) \}
	\end{equation}
	satisfies  $\tau(u)<\hat{T}_\alpha(u)$. We use the following facts: $\partial_1 G(u,w)<0$
	for~$w$ large, $\partial_2 G(u,w)>0$ for~$w$ large, and $f(u,t)$ explodes as
	$t\uparrow \hat{T}_\alpha(u)$. Thus, $g$ from~\eqref{eq:def g} satisfies
	\begin{equation}\label{eq:g blowup}
    	\lim_{t\uparrow \hat{T}_\alpha(u)} g(t) = -\infty,
	\end{equation}
	and~$K^{(u)}$ satisfies $\lim_{t\uparrow \hat{T}_\alpha(u)} K^{(u)}(t) = +\infty$.
	We can therefore pick $\varepsilon>0$ such that
	\begin{equation*}
	  g(t)<0,\ K^{(u)}(t) > 0 \quad \text{for} \quad
	    \hat{T}_\alpha(u)-\varepsilon \leq t<\hat{T}_\alpha(u).
	\end{equation*}
	For $z\in[0,1]$ and any~$s,t<\hat{T}_\alpha(u)$ satisfying
	$\hat{T}_\alpha(u)-\varepsilon \leq s\leq t$, we have
	\[
	  s+(t-s)z \geq s \geq \hat{T}_\alpha(u)-\varepsilon.
	\]
	Using this observation in~\eqref{eq:def Phi}, we see from a straightforward induction proof
	that
	\[
	  Q_{\alpha,n}(t,s) > 0, \quad n\geq 1,\ \hat{T}_\alpha(u)-\varepsilon \leq s\leq t
	    < \hat{T}_\alpha(u).
	\]
	The same then holds for the resolvent kernel~\eqref{eq:res kernel},
	\begin{equation}\label{eq:res pos}
	  R_{\alpha}(t,s) > 0, \quad  \hat{T}_\alpha(u)-\varepsilon \leq s\leq t
	    < \hat{T}_\alpha(u).
	\end{equation}
	By~\eqref{eq:vie der}, we obtain
	\begin{equation}\label{eq:der f}
	  \partial_1 f(u,t) = g(t) + \int_{0}^{t-\varepsilon} R_{\alpha}(t,s)g(s)ds
	    +  \int_{t-\varepsilon}^t R_{\alpha}(t,s)g(s)ds.
	\end{equation}
	Now note that
	\begin{equation}\label{eq:2 int}
	  \Big| \int_{0}^{t-\varepsilon} R_{\alpha}(t,s)g(s)ds \Big| \ll
	    - \int_{t-\varepsilon}^t R_{\alpha}(t,s)g(s)ds \quad
	    \text{as}\ t\uparrow \hat{T}_\alpha(u),
	\end{equation}
	where the right-hand side is positive. Indeed, \eqref{eq:2 int}	follows from~\eqref{eq:g blowup} and~\eqref{eq:res pos}, as $g(s)$
	on the left-hand side of~\eqref{eq:2 int} is~$O(1)$.
	Thus, letting $t\uparrow \hat{T}_\alpha(u)$, we find that the negative terms
	$g(t)+\int_{t-\varepsilon}^t$ on the right-hand side of~\eqref{eq:der f} dominate.
	This completes the proof.
\end{proof}
\begin{lemma}\label{le:ana}
  Let $u\in\mathbb R$ and $0<t<\hat{T}_{\alpha}(u)$. Then $f(\cdot,t)$
  is analytic at~$u$.
\end{lemma}
\begin{proof}
    According to Section 3.1.1 in~\cite{Br17}, the solution can be constructed
    by successive iteration and continuation. We just show that the first iteration
    step leads to an analytic function, because the finitely many further steps needed
    to arrive at arbitrary $t<\hat{T}_{\alpha}(u)$ can be dealt with analogously.
    Define the iterates $f_0=0$
    and
    \[
      f_{n+1}(v,s) := \frac{1}{\Gamma(\alpha)}\int_0^s (s-\tau)^{\alpha-1}
        G(v,f_n(v,\tau))d\tau,\quad n\geq 0.
    \]
    On a sufficiently small time interval, $f_{n}(v,\cdot)$ converges uniformly to $f(v,\cdot)$,
    and the solution can then be continued by solving an updated integral equation and so
    on (see~\cite{Br17}), until we hit $\hat{T}_\alpha(v)$. Now fix~$u$ and~$t$ as in the statement
    of the lemma. For a sufficiently small open complex neighborhood $U\ni u$, it is easy to see
    that $t<\hat{T}_\alpha(v)$ holds for $v\in U$. Define
    \[
      \gamma := 1 \vee \sup_{v\in U} |v| \quad \text{and} \quad
        \eta := 1 \vee \frac{t^\alpha}{\Gamma(\alpha+1)}.
    \]
    Then there is $c\geq 1$ such that, for arbitrary $v\in U$ and $w\in \mathbb C$,
    \begin{align*}
       |G(v,w)| &\leq c\big( (|w|\vee 1)^2 \vee \gamma(|w|\vee 1) \vee \gamma^2\big)\\
       %&\leq c \big(\gamma(|w|\vee 1)^2  \vee \gamma^2\big) \\
       &\leq c \gamma^2 (|w|\vee 1)^2 =: \theta (|w|\vee 1)^2.
    \end{align*}
    By the definition of~$f_n$, a trivial inductive proof then shows that
    \begin{equation}\label{eq:maj}
      \sup_{\substack{v\in U\\s\in[0,t]}} |f_n(v,s)| \leq (\theta \eta)^{2^n-1},\quad n\geq 0.
    \end{equation}
    By a standard result on parameter integrals (Theorem IV.5.8 in~\cite{El09}),
    the bound~\eqref{eq:maj} implies that each function $f_n(\cdot,t)$
    is analytic in~$U$. From the bounds in Section 3.1.1 of~\cite{Br17}, it is very
    easy to see that the convergence $f_n(v,t)\to f(v,t)$ is locally uniform w.r.t.~$v$
    for fixed~$t$. It is well known (see Theorem~3.5.1 in~\cite{GrKr06}) that this implies that the limit function $f(\cdot,t)$
    is analytic. 
\end{proof}
\begin{lemma}\label{le:mon T}
    The function $u\mapsto \hat{T}_\alpha(u)$ increases for $u\leq0$ and decreases
    for $u\geq1$.
\end{lemma}
\begin{proof}
  Recall that $\hat{T}_\alpha(u)=\infty$ in cases~(C) and~(D), which include $u\in[0,1]$.
  For case~(A), the assertion follows
  from part~(i) of Lemma~\ref{le:mon}.  
  So let~$u$ satisfy case~(B), where again we assume
  w.l.o.g.\ that $u<0$. Suppose that $\hat{T}_\alpha(\cdot)$ does not increase. Then we can
  pick $u_0<0$ such that any left neighborhood of~$u_0$ contains a point~$u$
  with $\hat{T}_\alpha(u)>\hat{T}_\alpha(u_0)$. From the continuity of $\partial_1 f$
  (see Lemma~\ref{le:ana}), part~(ii) of Lemma~\ref{le:mon}, and the continuity of~$\tau$
  from~\eqref{eq:tau},
  there are $u_1<u_0$ satisfying $\hat{T}_\alpha(u_1)>\hat{T}_\alpha(u_0)$
  and $t_1<\hat{T}_\alpha(u_0)$ such that $\partial_1 f(u,t)<0$ in the
  rectangle
  \[
    \{ (u,t) : u_1 \leq u \leq u_0, t_1 \leq t < \hat{T}_\alpha(u_1) \}.
  \]
  Then, $\lim_{t\uparrow \hat{T}_\alpha(u_0)} f(u_0,t)=\infty$ implies that
  \begin{equation}\label{eq:u1 expl}
    \lim_{t\uparrow \hat{T}_\alpha(u_0)} f(u_1,t)=\infty,
  \end{equation}
  because the inequality $\partial_1 f(u,t)<0$
  shows that
  $f(u_1,\cdot)$ must explode at least as fast as $f(u_0,\cdot)$.
  But~\eqref{eq:u1 expl} is a contradiction to $\hat{T}_\alpha(u_1)>\hat{T}_\alpha(u_0)$.
\end{proof}
\begin{lemma}\label{le:cons}
  Let $u\in\mathbb R$ and $0<t<\hat{T}_{\alpha}(u)$. Then~\eqref{eq:mgf} holds,
  where (as above) $f=c_3 \psi$ and~$f(u,\cdot)$ is the solution of~\eqref{eq:vie}.
\end{lemma}
\begin{proof}
  We assume that $u<0$, as $u\geq0$ is handled analogously.
  By Lemma~\ref{le:mon T}, $u\mapsto \hat{T}_{\alpha}(u)$ increases.
  In this proof, we write $M(u,t)$ for the right-hand side of~\eqref{eq:mgf},
  and $\tilde{M}(u,t)=\mathbb{E}[e^{uX_t}]$ for the mgf.
  Now fix $u< 0$ and $0<t<\hat{T}_{\alpha}(u)$ such that $(u,t)$ has positive
  distance from the graph of the increasing function $\hat{T}_\alpha(\cdot)$.
  Clearly, it suffices to consider pairs $(u,t)$ with this property.
  By Lemma~\ref{le:valER}, there are~$v^-<v^+$ such that
  \begin{equation}\label{eq:MM}
    M(v,t)=\tilde{M}(v,t),\quad v^-< v <v^+.
  \end{equation}
  We now show that~\eqref{eq:MM} extends to $u\leq v \leq v^+$ by analytic continuation.
  From general results on characteristic functions (Theorems~II.5a and~II.5b in~\cite{Wi41}),
  $v\mapsto \tilde{M}(v,t)$ is analytic in a vertical strip $w^-<\mathrm{Re}(v)<w^+$ of the complex plane,
  and has a singularity at $v=w^-$. If we suppose that $w^->u$, then Lemma~\ref{le:ana}
  leads to a contradiction: The left-hand side of~\eqref{eq:MM} would then be analytic at $v=w^-$,
  and the right-hand side singular. This shows that~\eqref{eq:MM} can be extended to the
  left up to~$u$ by analytic continuation.
\end{proof}
%
%It is now clear that the explosion time $\hat T$ of the VIE~\eqref{eq:vie}
%and the explosion time $T^*$ of the rough Heston model agree.
The following theorem
completes the proof of Theorem~\ref{thm:finite}.
\begin{theorem}\label{thm:validity}
  Let $u\in\mathbb R$. Then $\hat{T}_\alpha(u)= T^*_\alpha(u)$,
  and~\eqref{eq:mgf} holds for $0<t<T^*_\alpha(u)$.
\end{theorem}
\begin{proof}
  In the light of Lemma~\ref{le:cons}, it only remains to show that
  $\hat{T}_\alpha(u)\geq T^*_\alpha(u)$. (Obviously, Lemma~\ref{le:cons}
  implies that $\hat{T}_\alpha(u)\leq T^*_\alpha(u)$.)
  But this is clear from the continuity of the map
  $t\mapsto \tilde{M}(u,t)=\mathbb{E}[e^{uX_t}]$ on the interval $(0,T^*_\alpha(u))$.
  This continuity follows from the continuity of~$t\mapsto X_t$, Doob's submartingale
  inequality, and dominated convergence.
\end{proof}

For later use (Section~\ref{se:comp mom}), we give the following alternative argument:
\begin{proof}[Another proof that $\hat{T}_\alpha(u)\geq T^*_\alpha(u)$]
  Let us suppose that
  there is~$u_0$ with $\hat{T}_\alpha(u_0) < T^*_\alpha(u_0)$.
  From Theorem~\ref{thm:lb}, it is easy to see that the continuous
  function $u\mapsto \hat{T}_\alpha(u)$ tends to $+\infty$
  as $u<0$ approaches the region where $\hat{T}_\alpha(u)=\infty$.
  Thus, there is $u_1>u_0$ with $\hat{T}_\alpha(u_0)<\hat{T}_\alpha(u_1)<T^*_\alpha(u_0)$.
  
  We have seen in Lemma~\ref{le:ana} that $u\mapsto f(u,t)$ is analytic
  for any fixed~$t$. But it is also analytic w.r.t.~$t$ for fixed~$u$:
  From Theorem~1 in Lubich~\cite{Lu83}, itself based on earlier work by Miller and
  Feldstein~\cite{MiFe71}, it follows that~$f(u,\cdot)$ is analytic on
  the whole interval $(0,\hat{T}_\alpha(u))$. By Hartogs's theorem
  (Theorem~1.2.5 in~\cite{Kr92}), we conclude that the bivariate function $f(\cdot,\cdot)$ is continuous. Thus, the blow-up of $f(u_1,\cdot)$ at $\hat{T}_\alpha(u_1)$
  implies that
  \begin{equation}\label{eq:blowup u}
    \lim_{u\downarrow u_1}M(u,\hat{T}_\alpha(u_1))=\infty.
  \end{equation}
  (Again, we write~$M$ for the right-hand side of~\eqref{eq:mgf}
  and~$\tilde{M}$  for the mgf.) By Lemma~\ref{le:cons}, $\tilde{M}(\cdot,\hat{T}_\alpha(u_1))$ also blows up there, and thus has a singularity at~$u_1$.
  Since $u_0<u_1$, we conclude from Corollary~II.1b in~\cite{Wi41}
  that $\tilde{M}(u_0,\hat{T}_\alpha(u_1))=\infty$. As $S=e^X$
  is a martingale, this implies that $\tilde{M}(u_0,t)=\infty$ for all
  $t\geq\hat{T}_\alpha(u_1)$. In particular, it contradicts
  $\hat{T}_\alpha(u_1)<T^*_\alpha(u_0)$.
\end{proof}

\section{Validity of the fractional Riccati equation for complex~$u$}\label{se:complex}

Although the focus of this paper is on \emph{real}~$u$, the mgf needs to be evaluated
at complex arguments when used for option pricing. The following result fully justifies
using the fractional Riccati equation~\eqref{eqIII:Ha1}, respectively the VIE~\eqref{eq:vie},
to do so. As above, we write $T^*_\alpha(u)$ for the moment explosion time of~$S$,
and $\hat{T}_\alpha(u)$ for the explosion time of the VIE~\eqref{eq:vie}.

\begin{theorem}\label{thm:compl val}
  Let $u\in\mathbb C$. Then $T^*_\alpha(u)=T^*_\alpha(\mathrm{Re}(u))$, and~\eqref{eq:mgf}
  holds for $0<t<T^*_\alpha(u)$.
\end{theorem}

\begin{lemma}\label{le:compl hat}
  Let $u\in \mathbb C$. Then $\hat{T}_\alpha(u) \geq T^*_\alpha(u)$.
\end{lemma}
\begin{proof}
  Suppose that $\hat{T}_\alpha(u) < T^*_\alpha(u)$. The VIE~\eqref{eq:vie} translates
  into a two-dimensional real VIE for $(\mathrm{Re}(f),\mathrm{Im}(f))$. As $\hat{T}_\alpha(u)<\infty$,
  we get from Theorem~12.1.1 in~\cite{GrLoSt90} that $(\mathrm{Re}(f),\mathrm{Im}(f))$ explodes
  as $t\uparrow \hat{T}_\alpha(u)$. This contradicts the continuity of $t\mapsto \mathbb{E}[e^{uX_t}]$,
  where the latter is shown as in the proof of Theorem~\ref{thm:validity}.
\end{proof}

\begin{proof}[Proof of Theorem~\ref{thm:compl val}]
  The first statement is clear from $|e^{uX_t}|= e^{\mathrm{Re}(u)X_t}$.
  Now let $t>0$ be arbitrary. As above, we write $\tilde{M}$ for the mgf and~$M$
  for the right-hand side of~\eqref{eq:mgf}. By Theorem~\ref{thm:validity},
  we have $M(v,t)=\tilde{M}(v,t)$ for~$v$ in the real interval
  \[
    I:=\{ v \in \mathbb R : T^*_\alpha(v) \geq t\}.
  \]
  The function $\tilde{M}(\cdot,t)$ is analytic on the strip
  \begin{equation}\label{eq:strip}
    \{ v \in \mathbb C : \mathrm{Re}(v) \in I\}
      = \{ v \in \mathbb C : T^*_\alpha(v) \geq t\}.
  \end{equation}
  By the same argument as in Lemma~\ref{le:ana}, the function $M(\cdot,t)$ is analytic
  on the set $\{ v \in \mathbb C : \hat{T}_\alpha(v) \geq t \}$, which contains
  the strip~\eqref{eq:strip} by Lemma~\ref{le:compl hat}. Therefore, $M(\cdot,t)$
  and $\tilde{M}(\cdot,t)$ agree on~\eqref{eq:strip} by analytic continuation.
  This implies the assertion.
\end{proof}

\section{Computing the explosion time}\label{se:comp exp}

Recall that, for fixed $u\in \mathbb R$, the explosion time $T^*_\alpha(u)$
of the rough Heston model
is the blow-up time of $f(t)=f(u,t)=c_3 \psi(u,t)$, where~$\psi$ solves the fractional
Riccati initial value problem~\eqref{eqIII:Ha1}--\eqref{eqIII:Ha2}.
We know from Theorem~\ref{thm:finite} that $T^*_\alpha(u)<\infty$ exactly
in the cases~(A) and~(B), defined in Section~\ref{se:prelim}.
We now develop a method (Algorithm~\ref{algo:T}) to compute $T^*_\alpha(u)$ for $u$ satisfying
the conditions of case~(A). In case~(B), a lower bound can be computed,
which is sometimes sharper than the explicit bound~\eqref{eqIII:T*-lower}.
The function~$f$ satisfies the fractional Riccati equation
\begin{equation}\label{eq:ric2}
  D^\alpha f = d_1 + d_2 f + f^2,
\end{equation}
where $d_1(u):=c_1(u) c_3$ and $d_2(u):=c_2(u)$, with
initial condition $I^{1-\alpha}f(0)=0$.
(Recall that we often suppress the dependence on~$u$ in the notation.)
We try a fractional power series ansatz
\begin{equation}\label{eq:fps}
  f(t) \stackrel{!}{=} \sum_{n=1}^\infty a_n(u) t^{\alpha n}
\end{equation}
with unknown coefficients $a_n=a_n(u)$.
\begin{lemma} (see e.g.~\cite{KiSrTr06})
  Let $\alpha\in(0,1)$.
   The fractional integral and derivative of power functions are given by
    \begin{align}
      I^\alpha_t t^\nu &= t^{\nu+\alpha}\frac{\Gamma(\nu+1)}{\Gamma(\nu+\alpha+1)}\qquad\text{for~}\nu>-1, \label{eqIII:frac1a} \\
      D^\alpha_t t^\nu &= t^{\nu-\alpha}\frac{\Gamma(\nu+1)}{\Gamma(\nu-\alpha+1)}\qquad\text{for~}\nu>-1+\alpha. \label{eqIII:frac1b}
    \end{align}
\end{lemma}
%
%The fractional integral and derivative of a power look as follows:
%\begin{align}
%  I^\alpha t^\nu &= \frac{\Gamma(\nu+1)}{\Gamma(\nu+\alpha+1)} t^{\nu+\alpha},
%    \quad \nu>-1, \label{eq:int pow}\\
%  D^\alpha t^\nu &= \frac{\Gamma(\nu+1)}{\Gamma(\nu-\alpha+1)} t^{\nu-\alpha},
%    \quad \nu>\alpha-1. \label{eq:der pow}
%\end{align}
By~\eqref{eqIII:frac1a}, the fractional power series~\eqref{eq:fps} (formally) satisfies
the initial condition~\eqref{eqIII:Ha2}. Inserting~\eqref{eq:fps} into~\eqref{eq:ric2}
and using~\eqref{eqIII:frac1b}, we obtain
\begin{align}
  \sum_{n=0}^\infty a_{n+1} v_{n+1} t^{\alpha n} &= d_1 + \sum_{n=1}^\infty d_2 a_n t^{\alpha n}
    + \sum_{n=2}^\infty \Big( \sum_{k=1}^{n-1} a_k a_{n-k} \Big) t^{\alpha n} \notag \\
   &= d_1 + d_2 a_1 t^\alpha + \sum_{n=2}^\infty \Big( d_2 a_n
     + \sum_{k=1}^{n-1} a_k a_{n-k} \Big )t^{\alpha n}, \label{eq:ins ans}
\end{align}
where
\[
  v_n := \frac{\Gamma(\alpha n+1)}{\Gamma(\alpha n-\alpha+1)}.
\]
Note that~$v_n$ is an increasing sequence; this follows easily from the fact that
$\log \circ\, \Gamma$ is convex (see Example~11.14 in~\cite{Sc05}).
By Stirling's formula, $v_n\sim(\alpha n)^\alpha$
for $n\to\infty$.
%It should be easy to show that $v_n\geq (\alpha n)^\alpha$.
From~\eqref{eq:ins ans}, we obtain a convolution recurrence for $a_n=a_n(u)$:
\begin{align}
  a_1(u) &= d_1(u)/v_1, \label{eq:rec ic} \\
 % a_2 &= d_2 a_1/v_2, \notag \\
  a_{n+1}(u) &= \frac{1}{v_{n+1}}\Big( d_2(u) a_n(u)
    + \sum_{k=1}^{n-1}a_k(u) a_{n-k}(u)\Big), \quad n\geq1.
  \label{eq:rec}
\end{align}
The function $f$ can thus be expressed as $f(u,t) = F(u,t^\alpha)$, where
\[
   F(u,z):=\sum_{n=1}^\infty a_n(u) z^n.
\]
\begin{lemma}\label{le:rad}
Let $u\in\mathbb{R}$, satisfying case (A)
(recall the definition in Section~\ref{se:prelim}). Then $F(u,\cdot)$
is analytic at zero, with a positive and finite radius of convergence~$R(u)$.
\end{lemma}
\begin{proof}
  To see that the radius of convergence is positive, we show that there is $A=A(u)>0$ such that
  \begin{equation}\label{eq:A}
    |a_n| \leq A^n n^{\alpha-1}, \quad n\geq1.
  \end{equation}
  (Adding the factor $n^{\alpha-1}$ to this geometric bound facilitates the inductive proof.)
  We have
  %\[
  %  \frac{|d_2|/n+n^{\alpha-1}}{\alpha^\alpha(1+1/n)^{\alpha-1} n^{\alpha-1}} \to \alpha^{-\alpha},
  %  \quad n\to\infty.
  %\]
  \[
    \frac{\alpha^{-\alpha}|d_2|n^{-1}}{(n+1)^{\alpha-1}}
    + \frac{2\alpha^{-\alpha}\Gamma(\alpha)^2 n^{\alpha-1}}
      {\Gamma(2\alpha)(n+1)^{\alpha-1}} \to \frac{2\alpha^{-\alpha}\Gamma(\alpha)^2}
      {\Gamma(2\alpha)},
    \quad n\to\infty.
  \]
  Choose~$n_0$ such that the left-hand side is
  $\leq 3\alpha^{-\alpha}\Gamma(\alpha)^2/\Gamma(2\alpha)$ for all
  $n\geq n_0$, and such that $2v_n\geq (\alpha n)^\alpha$ for all
  $n\geq n_0$.
  The latter is possible because $v_n\sim(\alpha n)^\alpha$.
  Fix a number $A$ with
  $A\geq 3\alpha^{-\alpha}\Gamma(\alpha)^2/\Gamma(2\alpha)$ and such that $A^n n^{\alpha-1}\geq |a_n|$ holds
  for $1\leq n\leq n_0$. Let $n\geq n_0$ and assume, inductively, that
  $|a_k| \leq A^k k^{\alpha-1}$ holds for $1\leq k\leq n$. From the recurrence~\eqref{eq:rec},
  we then obtain
  \begin{align*}
    |a_{n+1}|&\leq 2(\alpha n+\alpha)^{-\alpha}\left(|d_2| A^n n^{\alpha-1}
    +A^n \sum_{k=1}^{n-1}k^{\alpha-1}(n-k)^{\alpha-1}\right) \\
    &\leq 2(\alpha n)^{-\alpha}\left(|d_2| A^n n^{\alpha-1}
    +A^n \sum_{k=1}^{n-1}k^{\alpha-1}(n-k)^{\alpha-1}\right).
  \end{align*}
  Since $x^{\alpha-1}(n-x)^{\alpha-1}$ is a strictly convex function
  of~$x$ on $(0,n)$ with minimum at $n/2$, it
  is easy to see that
  \begin{align*}
    \sum_{k=1}^{n-1}k^{\alpha-1}(n-k)^{\alpha-1}
    &\leq \int_0^n x^{\alpha-1}(n-x)^{\alpha-1} dx \\
    &= n^{2\alpha-1}\int_0^1 y^{\alpha-1}(1-y)^{\alpha-1} \\
     &=       \frac{\Gamma(\alpha)^2}{\Gamma(2\alpha)}n^{2\alpha-1},
  \end{align*}
  where the last equality follows from the well-known representation of the
  beta function in terms of the gamma function (see 12.41 in~\cite{WhWa96}).
  We conclude
  \begin{align*}
    |a_{n+1}|&\leq 2\alpha^{-\alpha} |d_2| A^n n^{-1}
    + \frac{2\alpha^{-\alpha}\Gamma(\alpha)^2}{\Gamma(2\alpha)} A^n n^{\alpha-1} \\
    &= A^n (n+1)^{\alpha-1} \left(
     \frac{2\alpha^{-\alpha}|d_2|n^{-1}}{(n+1)^{\alpha-1}}
    + \frac{2\alpha^{-\alpha}\Gamma(\alpha)^2 n^{\alpha-1}}
      {\Gamma(2\alpha)(n+1)^{\alpha-1}}
       \right) \\
    &\leq A^n (n+1)^{\alpha-1}   \frac{3\alpha^{-\alpha}\Gamma(\alpha)^2}{\Gamma(2\alpha)} \\
    &\leq A^{n+1} (n+1)^{\alpha-1}.
  \end{align*}
  This completes the inductive proof of~\eqref{eq:A}.
  
  The finiteness of the radius of convergence will follow from the existence of a
  number $B=B(u)>0$ such that
  \begin{equation}\label{eq:B bd}
    a_n\geq B^n, \quad n\geq1.
  \end{equation}
  To this end, define
  \[
     r_n := \frac{d_2+n-1}{v_{n+1}}, \quad n\geq 1.
  \]
  By Stirling's formula, we have $r_n/r_{n-1} = 1 +(1-\alpha)/n + O(n^{-2})$ as
  $n\to\infty$, and so~$r_n$ eventually increases. Let $n_0\geq2$ be such that~$r_n$
  increases for $n\geq n_0$, and define
  \[
    B := \min\{r_{n_0},a_1,a_2^{1/2},\dots,a_{n_0}^{1/{n_0}} \}.
  \]
  This number satisfies $a_n\geq B^n$ for $n\leq n_0$ by definition. Let us fix some $n\geq n_0$
  and assume, inductively, that $a_k\geq B^k$ holds for $1\leq k\leq n$. By~\eqref{eq:rec}
  \begin{align*}
    a_{n+1} &\geq \frac{1}{v_{n+1}}(d_2 B^n + (n-1)B^n) \\
    &= B^n r_n \\
    &\geq B^n r_{n_0} \geq B^{n+1}.
  \end{align*}
  Thus,~\eqref{eq:B bd} is proved by induction.
\end{proof}
{}From the estimates in Lemma~\ref{le:rad}, it is clear that termwise
fractional derivation of the series~\eqref{eq:fps} is allowed, and so
the right-hand side of~\eqref{eq:fps} really represents the solution~$f$
of~\eqref{eq:ric2} with initial condition $I^{1-\alpha}f(0)=0$, as long
as~$t$ satisfies $0\leq t<R(u)^{1/\alpha}$.
We proceed to show how the explosion time $T_\alpha^*(u)$ can be computed from the coefficients~$a_n(u)$.
The essential fact is that there is no gap between $R(u)^{1/\alpha}$
and $T_\alpha^*(u)$.
For this, we require the following classical result from complex analysis (\cite{Re91},
p.~235).

\begin{theorem}[Pringsheim's theorem, 1894]\label{thm:pring}
  Suppose that the power series $F(z)=\sum_{n=0}^\infty a_n z^n$ has positive finite radius
  of convergence~$R$, and that all the coefficients
  are non-negative real numbers. Then~$F$ has a singularity at~$R$.
\end{theorem}

\begin{theorem}\label{thm:conv}
  Suppose that $u\in\mathbb R$ satisfies case~(A). Define the sequence $a_n(u)$ by the
  recurrence~\eqref{eq:rec} with initial value~\eqref{eq:rec ic}. Then we have
  \begin{equation}\label{eq:lim}
    \limsup_{n\to\infty}  a_n(u)^{-1/(\alpha n)} = T_\alpha^*(u).
   %   \limsup_{n\to\infty} \bigg( a_n(u) n^{1-\alpha}
   %   \frac{\Gamma(\alpha)^2}{\alpha^\alpha \Gamma(2\alpha)} \bigg)^{-1/(\alpha n)} = T_\alpha^*(u).
  \end{equation}
\end{theorem}
\begin{proof}
  %The subexponential factor $n^{1-\alpha}\frac{\Gamma(\alpha)^2}{\alpha^\alpha \Gamma(2\alpha)}$
  %is of no relevance for the proof; its significance will be explained below.
  %
  Recall that~$f(u,\cdot)$, the solution of~\eqref{eq:ric2}, also solves the Volterra integral
  equation~\eqref{eq:vie}. From the references on smoothness cited
  before~\eqref{eq:blowup u}, it follows that~$f(u,\cdot)$ is analytic on
  the whole interval $(0,T_\alpha^*(u))$.
  As $f(u,t)$ blows up for $t\uparrow T_\alpha^*(u)$ by Proposition~\ref{prop:case1},
  and $t\mapsto F(u,t^\alpha)$ is
  analytic on $(0,R(u)^{1/\alpha})$, we must have $R(u)^{1/\alpha} \leq T_\alpha^*(u)$.
  
  Assume for contradiction that $R(u)^{1/\alpha} < T_\alpha^*(u)$. Then~$f$ is
  analytic at $R(u)^{1/\alpha}$.
  But since $z\mapsto z^{1/\alpha}$ is analytic at~$R(u)>0$, the composition $F(u,z)=f(u,z^{1/\alpha})$
  would be analytic at $z=R(u)$ as well, which contradicts Theorem~\ref{thm:pring}.
  Therefore,
  \[
    R(u)^{1/\alpha} = T_\alpha^*(u).
  \]
  It is well-known that the radius of convergence is given by the Cauchy-Hadamard
  formula~\cite[p.~111]{Re91}
  \[
    R(u)^{-1} = \limsup_{n\to\infty} a_n(u)^{1/n},
  \]
  which concludes the proof.
\end{proof}
Note that, in case~(B), we can argue similarly as in the preceding proof. However,
the coefficients~$a_n$ are no longer positive, and so Pringsheim's theorem is not applicable.
Then, the inequality $R(u)^{1/\alpha}\leq T_\alpha^*(u)$ need not be an equality. Still, we can compute
a lower bound for the explosion time:
\begin{equation}\label{eq:lb T}
  \limsup_{n\to\infty} |a_n(u)|^{-1/(\alpha n)} \leq T_\alpha^*(u).
\end{equation}
Now assume that we are in case~(A) again. We now discuss how to speed up the convergence
in~\eqref{eq:lim}.
Roberts and Olmstead~\cite{RoOl96} studied the blow-up behavior of solutions of nonlinear
Volterra integral equations with (asymptotically) fractional kernel.
Their arguments hinge on the asymptotic behavior
of the nonlinearity for large argument. In particular, in our situation,
with $G(u,w)$ from~\eqref{eqIII:G1} satisfying $G(u,w)\sim w^2$ for $w\to\infty$, formula~(3.2) in~\cite{RoOl96} yields
\begin{equation}\label{eq:Rob}
  f(t) \stackrel{(?)}{\sim} \frac{\Gamma(2\alpha)}{\Gamma(\alpha)}%T_\alpha^*(u)^\alpha
    (T_\alpha^*(u)-t)^{-\alpha}, \quad t\uparrow T_\alpha^*(u).
\end{equation}
We write $\stackrel{(?)}{\sim}$ for two reasons: First, our integral equation~\eqref{eq:vie} does
not quite satisfy the technical assumptions in~\cite{RoOl96}. %, except in case~(A) for $\rho>0$.
Second, not all steps in~\cite{RoOl96} are rigorous.
We proceed, heuristically, to infer refined asymptotics of~$a_n(u)$ from~\eqref{eq:Rob}. Define
\[
  \Phi(z):= \sum_{n=1}^\infty a_n(u) R(u)^n z^n,
\]
a power series with radius of convergence~$1$, by the definition of~$R(u)$ in Lemma~\ref{le:rad}.
Its asymptotics for $z\uparrow1$ can be derived from~\eqref{eq:Rob}. Recall that the explosion
time and the radius of convergence of~$F$ are related by $T_\alpha^*(u)=R(u)^{1/\alpha}$.
\begin{align*}
  \Phi(z) &= f\big((Rz)^{1/\alpha}\big) \\
  &\stackrel{(?)}{\sim} \frac{\Gamma(2\alpha)}{\Gamma(\alpha)}
    (T_\alpha^*-(Rz)^{1/\alpha})^{-\alpha} \\
  &= \frac{\Gamma(2\alpha)}{\Gamma(\alpha)} R^{-1} (1-z^{1/\alpha})^{-\alpha} \\
  &\sim \frac{\alpha^\alpha\Gamma(2\alpha)}{\Gamma(\alpha)} R^{-1}(1-z)^{-\alpha},
    \quad z\uparrow 1.
\end{align*}
The method of \emph{singularity analysis} (see Section~VI in~\cite{FlSe09}) allows to transfer the asymptotics of~$\Phi$ to
asymptotics of its Taylor coefficients $a_n R^n$.
Sweeping some analytic conditions under the rug, we arrive at
\[
  a_n(u)R(u)^n \stackrel{(?)}{\sim}  \frac{\alpha^\alpha\Gamma(2\alpha)}{\Gamma(\alpha)}\,
  R^{-1} \frac{n^{\alpha-1}}{\Gamma(\alpha)}, \quad n\to\infty,
\]
and thus
\begin{equation}\label{eq:asym a}
  a_n(u) \stackrel{(?)}{\sim} R(u)^{-n-1} n^{\alpha-1}
      \frac{\alpha^\alpha \Gamma(2\alpha)}{\Gamma(\alpha)^2}, \quad n\to\infty.
\end{equation}
Numerical tests confirm~\eqref{eq:asym a}, and we have little doubt that it is true
(in case~(A)).
Summing up, $T_\alpha^*(u)$ can be computed by the following algorithm,
which converges much faster
than the simpler approximation $\limsup_{n\to\infty} a_n^{-1/(\alpha n)}$:
\begin{algorithm}\label{algo:T}
Let $u$ be a real number satisfying case~(A).
\begin{itemize}
  \item Fix $n_{\max} \in \mathbb N$ (e.g.\ $n_{\max}=100$),
  \item compute $a_1(u),\dots,a_{n_{\max}}(u)$ by the recursion~\eqref{eq:rec},
  \item compute the approximation
     \begin{equation}\label{eq:algo}
      \bigg( a_n(u) n^{1-\alpha}
      \frac{\Gamma(\alpha)^2}{\alpha^\alpha \Gamma(2\alpha)} \bigg)^{\frac{-1}{\alpha (n+1)}}
      \bigg|_{n=n_{\max}}
      \approx T_\alpha^*(u)
     \end{equation}
     for the explosion time.
\end{itemize}
\end{algorithm}
We stress that, while the arguments leading to~\eqref{eq:asym a} are heuristic,
we have rigorously shown in Theorem~\ref{thm:conv} that $T_\alpha^*(u)$
is the $\limsup$ of the left-hand side of~\eqref{eq:algo}.
The heuristic part is that the subexponential factor $n^{1-\alpha}\times \mathit{const}$
improves the relative error of the approximation from $O(\frac{\log n}{n})$
to $O(\frac{1}{n^2})$.
Note that our approach to compute the blow-up time can of course be extended to more general fractional
Riccati equations.
Finally, as mentioned above (see~\eqref{eq:lb T}), we can compute a \emph{lower bound}
for $T_\alpha^*(u)$ if it is finite, but~$u$ is outside of case~(A):
\begin{algorithm}\label{algo:lb}
Let $u$ be a real number satisfying case~(B).
\begin{itemize}
  \item Fix $n_{\max} \in \mathbb N$ (e.g.\ $n_{\max}=200$),
  \item compute $a_1(u),\dots,a_{n_{\max}}(u)$ by the recursion~\eqref{eq:rec},
  \item compute the approximate lower bound
     \[
       |a_n(u)|^{-1/(\alpha n)}\big|_{n=n_{\max}}\, \lessapprox T_\alpha^*(u)
     \]
     of the explosion time.
\end{itemize}
\end{algorithm}
\begin{remark}
As for the applicability of Algorithm~\ref{algo:T}, suppose that $\rho<0$
(with analogous comments applying to the less common case $\rho>0$).
From~\eqref{eqIII:e_0}, we have
 $e_0(u)\sim\tfrac12 \rho\xi u>0$ for $u\downarrow -\infty$, and so we are
in case~(A) for large enough~$|u|$. More precisely, case~(A) corresponds
to the interval $u\in(-\infty,\lambda/(\xi \rho)]$. For~$u$ from that interval,
the explosion time can be computed by Algorithm~\ref{algo:T}. To the right
of $u=\lambda/(\xi \rho)$, there is a (possibly empty) interval corresponding
to case~(B), where $T_\alpha^*(u)$ is still finite, but Algorithm~\ref{algo:T}
cannot be applied. Still, a lower bound can be computed by~\eqref{eq:lb T},
and we have the bounds from Theorems~\ref{thm:lb} and~\ref{thmIII:T*-upper},
which can be easily evaluated numerically.
Proceeding further to the right on the $u$-axis, we encounter an interval
containing $[0,1]$, on which $T_\alpha^*(u)=\infty$ (cases~(C) and~(D)). Afterwards, $T_\alpha^*(u)$
becomes finite again, but these~$u$ belong to case~(B), leaving
us with bounds for $T_\alpha^*(u)$ only.
\end{remark}

\begin{figure}[ht]
	\centering
  \includegraphics[width=0.7\linewidth]{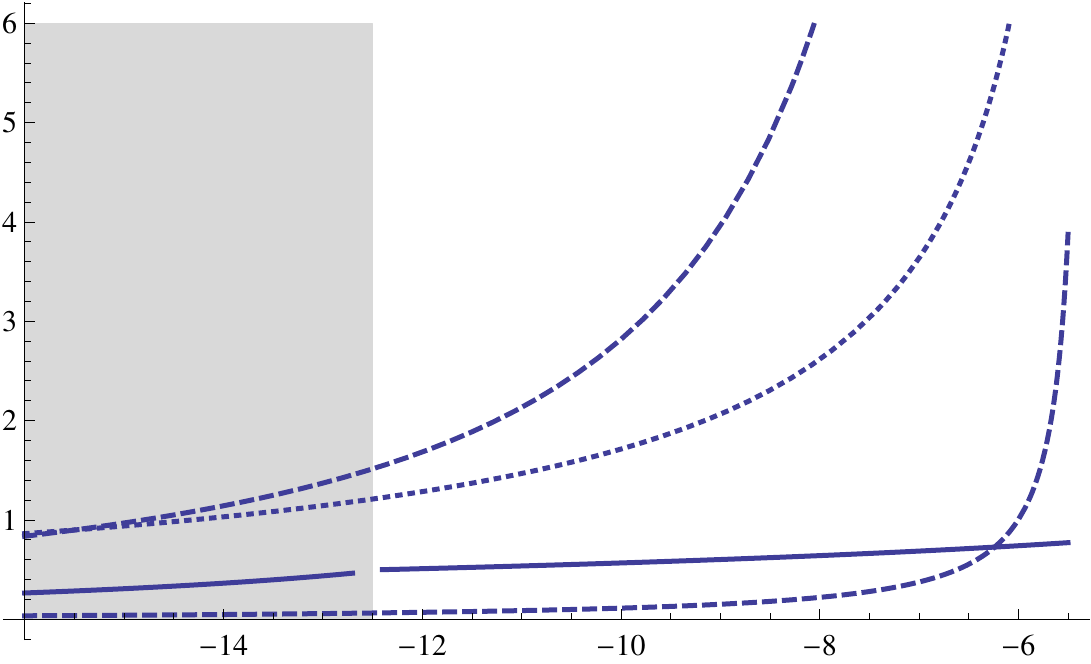} 
	\caption{Moment explosion time and bounds for $u<0$. The parameters are
	$\alpha=0.6,$ $\rho=-0.8,$ $\lambda=2,$ and $\xi=0.2.$
	The gray rectangle sits above the interval $(-\infty,\lambda/(\xi\rho)]=
	(-\infty,-12.5]$ corresponding to case~(A).
	Left solid curve: $T^*_\alpha,$
	computed by Algorithm~\ref{algo:T}. Right solid curve: lower bound, computed
	by Algorithm~\ref{algo:lb}. Dashed curves: bounds from Theorems~\ref{thm:lb}
	resp.~\ref{thmIII:T*-upper}. Dotted curve: $T^*_1$ (classical Heston model).}
	\label{fig:06}
\end{figure}

\begin{figure}[ht]
	\centering
  \includegraphics[width=0.7\linewidth]{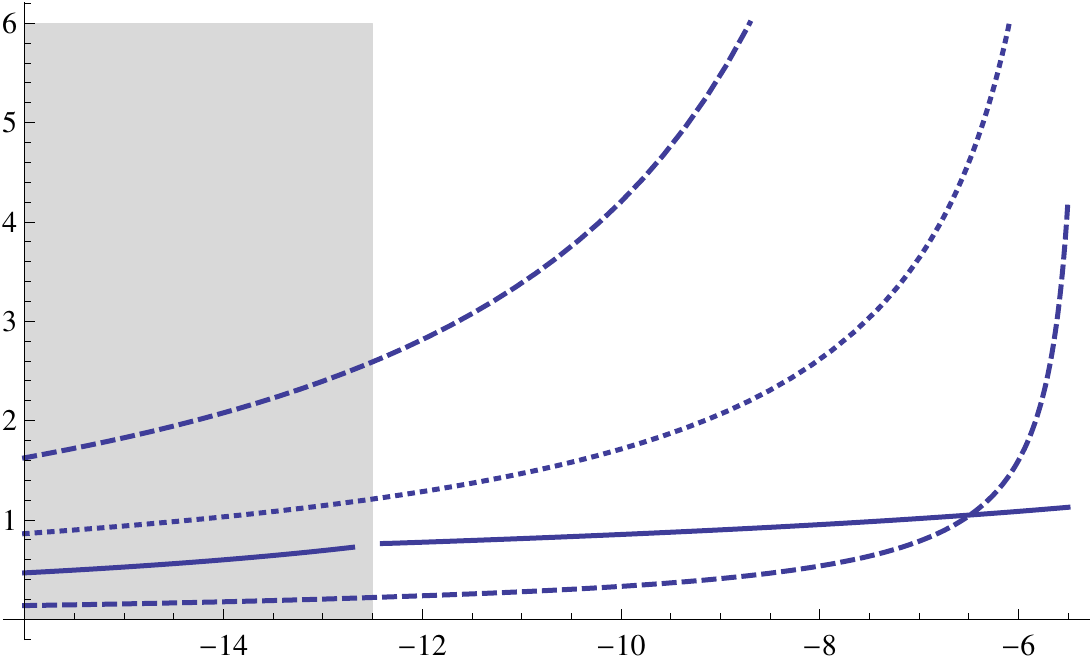} 
	\caption{As Figure~\ref{fig:06}, but with $\alpha=0.75$.}
  \label{fig:075}
\end{figure}

\begin{figure}[ht]
	\centering
  \includegraphics[width=0.7\linewidth]{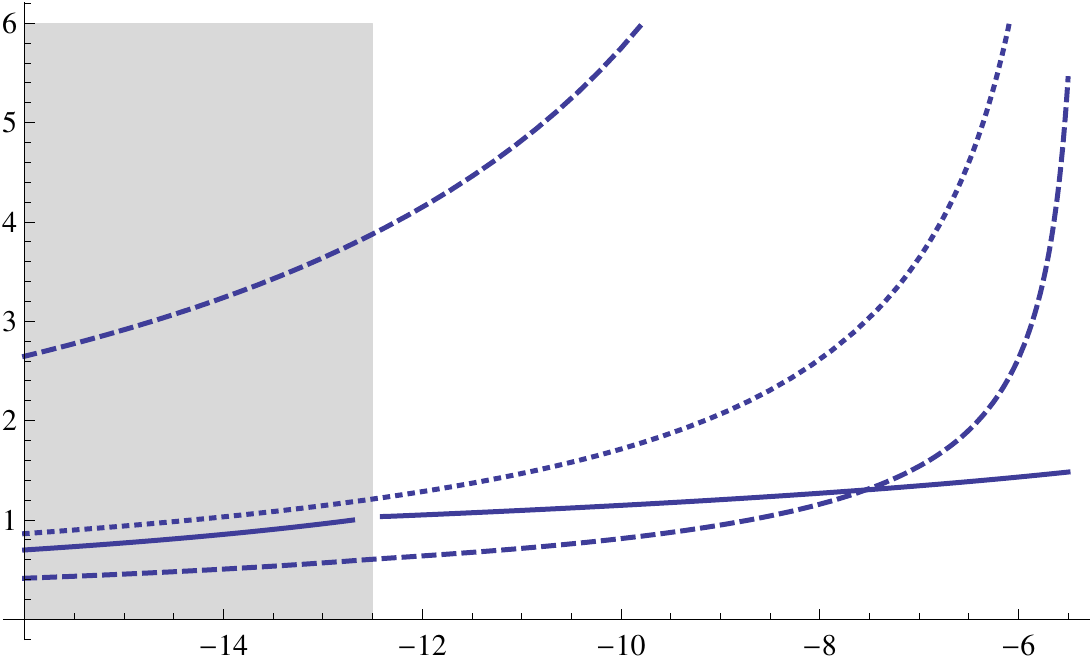} 
	\caption{As Figure~\ref{fig:06}, but with $\alpha=0.9$. Note that $T^*_\alpha$
	(left solid curve) is close to $T^*_1$ (classical Heston model, i.e.\ $\alpha=1$; dotted curve).}
   \label{fig:09}
\end{figure}

To conclude this section, we note that~$f$ can be approximated by replacing
the coefficients in~\eqref{eq:fps} by the right-hand side of~\eqref{eq:asym a}.
Let us write~$b_n(u)$ for the latter. Retaining the first~$N$ \emph{exact}
coefficients, this leads to the approximation
\begin{align}
  f(u,t) &\approx \sum_{n=1}^\infty b_n(u) t^{\alpha n} 
     + \sum_{n=1}^{N}\big(a_n(u)-b_n(u)\big) t^{\alpha n} \notag \\
  &=  \frac{\alpha^\alpha \Gamma(2\alpha)}{\Gamma(\alpha)^2 T_\alpha^*(u)^\alpha}
    \mathrm{Li}_{1-\alpha}\big( (t/T_\alpha^*(u))^\alpha \big) 
    + \sum_{n=1}^{N}\big(a_n(u)-b_n(u)\big) t^{\alpha n}, \label{eq:Li}
\end{align}
where $\mathrm{Li}_{\nu}(z):=\sum_{n=1}^\infty z^n/n^\nu$ denotes the polylogarithm.
While this approximation seems to be very accurate even for small~$N$ (see~\cite{Ge18}), it is limited
to \emph{real}~$u$ satisfying case~(A), and thus not applicable to option pricing.

\section{Finiteness of the critical moments}\label{se:cm fin}

While we have analyzed the \emph{explosion time} of the rough Heston model so far,
in most applications of moment explosions (see the introduction), the
\emph{critical moments}
\begin{align}
  u^+(T) &:= \sup\{u\in\mathbb{R}\colon\mathbb{E}[e^{u X_T}]<\infty\}, \notag \\
  u^-(T) &:= \inf\{u\in\mathbb{R}\colon\mathbb{E}[e^{u X_T}]<\infty\}, \quad T>0,\label{eq:def u-}
\end{align}
are of interest.
Using the upper bound for the moment explosion time $T_\alpha^\ast$ in Theorem~\ref{thmIII:T*-upper}, we
will now show the finiteness of the critical moments for every maturity $T>0$.
Computing $u^+(T)$ and $u^-(T)$ is discussed in Section~\ref{se:comp mom}.

%-----------------------------------------------------------------------------------------------------

\begin{theorem}
  In the rough Heston model the critical moments $u^+(T)$ and $u^-(T)$ are finite for every $T>0$.
\end{theorem}

%-----------------------------------------------------------------------------------------------------

\begin{proof}
  Only the finiteness of $u^+(T)$ is proven, as the proof for $u^-(T)$ is very similar. Denote the upper bound of $T_\alpha^\ast(u)$ in \eqref{eqIII:T*-upper} by $B(u)$ for all $u\in\mathbb{R}$ in the cases~\ref{itemIII: A} and~\ref{itemIII: B}. First, we show that for sufficiently large $u$, we are always in case~\ref{itemIII: A} or~\ref{itemIII: B}, depending on the sign of the correlation parameter $\rho$. From \eqref{eqIII:e_0} and \eqref{eqIII:e_1}, it is easy to see that
  \begin{equation}\label{eqIII:e0-e1}
    e_0(u)\sim \tfrac12\xi\rho u \quad\text{and}\quad e_1(u)\sim -\tfrac14\xi^2\bar\rho^2 u^2 \quad\text{as~} u\to\infty,
  \end{equation}
  where $\bar\rho^2 = 1- \rho^2$. Thus, eventually $e_1(u)<0$ for sufficiently large~$u$. In the next step, we show that the upper bound $B(u)$ converges to $0$ as $u\to\infty$. Indeed, in case~\ref{itemIII: A} the integral in \eqref{eqIII:T*-upper} satisfies
  \begin{align*}
    \int_0^\infty \bigg(\frac{w}{G(u,w)}\bigg)^{1/\alpha}\frac{dw}{w} &= \int_0^{u} \bigg(\frac{w}{G(u,w)}\bigg)^{1/\alpha}\frac{dw}{w} + \int_{u}^\infty \bigg(\frac{w}{G(u,w)}\bigg)^{1/\alpha}\frac{dw}{w} \\
    &\leq G(u,0)^{-1/\alpha}\int_0^{u} w^{-1+1/\alpha}\,dw +\int_{u}^\infty w^{-1-1/\alpha}\,dw \\
    &\leq c u^{-1/\alpha}, \quad u\to\infty,
  \end{align*}
  for some $c>0$, using the monotonicity $G(u,w)\geq G(u,0)$, the inequality $G(u,w)\geq w^2$ and $G(u,0)=c_3c_1(u)\sim \tfrac12 c_3 u^2$ as $u\to\infty$.
  
  If we are eventually in case~\ref{itemIII: B} as $u\to\infty$, then $-e_1(u)>0$ and $-e_0(u)>0$ holds for all sufficiently large $u$. Note that in this case $G(u,\cdot)$ attains its global minimum at $-e_0(u)$ and the minimum value is $-e_1(u)$. Thus, the integral  in \eqref{eqIII:T*-upper} satisfies
  \begin{align*}
    \int_0^\infty &\bigg(\frac{w}{G(u,w)}\bigg)^{1/\alpha}\frac{dw}{w} = \int_0^{-2e_0(u)} \bigg(\frac{w}{G(u,w)}\bigg)^{1/\alpha}\frac{dw}{w} + \int_{-2e_0(u)}^\infty \bigg(\frac{w}{G(u,w)}\bigg)^{1/\alpha}\frac{dw}{w} \\
    &\leq (-e_1(u))^{-1/\alpha}\int_0^{-2e_0(u)} w^{-1+1/\alpha}\,dw +4^{1/\alpha}\int_{-2e_0(u)}^\infty w^{-1-1/\alpha}\,dw \\
    &= \alpha\left((-e_1(u))^{-1/\alpha}(-2e_0(u))^{1/\alpha} + 4^{1/\alpha}(-2e_0(u))^{-1/\alpha}\right) \\
    &\leq c u^{-1/\alpha}, \quad u\to\infty,
  \end{align*}
  for some $c>0$, using the monotonicity $G(u,w)\geq -e_1(u)$, the inequality $G(u,w)\geq (w+e_0(u))^2\geq w^2/4$ on $[-2e_0(u),\infty)$ and \eqref{eqIII:e0-e1}.
  
  Altogether, we have $\lim_{u\to\infty} B(u)=0$. Since $0\leq T_\alpha^\ast(u)\leq B(u)$, the same is true for the moment explosion time $T_\alpha^\ast$, i.e.\ $\lim_{u\to\infty} T_\alpha^\ast(u) = 0$.
  Now  let $T>0$ be arbitrary. Then there exists $u_0\in\mathbb{R}$ such that $T_\alpha^\ast(u)<T$ for all $u\geq u_0$. This inequality implies $\mathbb{E}[e^{uX_T}]=\infty$ for all $u\geq u_0$, and therefore $u_+(T)\leq u_0$.
\end{proof}

From the preceding proof, it easily follows that $u^+(T)$ and $u^-(T)$ are of order $T^{-\alpha}$ as $T\downarrow0$. This is consistent with the classical Heston model ($\alpha=1$),
where the decay order is $T^{-1}$, by inverting~\eqref{eq:He et}.

\section{Computing the critical moments}\label{se:comp mom}

We first collect some simple facts that apparently have not been made explicit
in the literature on moment explosions. Moment explosion time and critical moments
are defined as in~\eqref{eq:def exp time} resp.~\eqref{eq:def u-}.
\begin{lemma}\label{le:gen}
  Let $S=(S_t)_{t\geq0}=(e^{X_t})_{t\geq0}$ be a positive stochastic process.
  Its moment explosion time is denoted by $T^*(u)$, $u\in\mathbb R$, and
  its critical moments by $u^-(T)$ and $u^+(T)$, $T>0$.
  \begin{itemize}
    \item[(i)] $T^*(u)$ increases for $u\leq0$, and decreases for $u\geq 0$.
    \item[(ii)] If~$S$ is a martingale, then $u^+(T)$ decreases, and $u^-(T)$ increases.
    \item[(iii)] Suppose that~$S$ is a martingale. If $T^*(u)$ decreases
    strictly on the interval
    \[
      \mathcal{D}^+:= \{ u\geq 1 : T^*(u)<\infty\},
    \]
    then $T^*=(u^+)^{-1}$ on~$\mathcal{D}^+$. Analogously, if $T^*(u)$ increases
    strictly on the interval
    \[
      \mathcal{D}^-:= \{ u\leq 0 : T^*(u)<\infty\},
    \]
    then $T^*=(u^-)^{-1}$ on~$\mathcal{D}^-$.
  \end{itemize}
\end{lemma}
\begin{proof}
  (i) As $T^*(u)=\infty$, we may assume that $u>0$ ($u<0$ is analogous). The assertion
  follows from Jensen's inequality, since $x\mapsto x^{v/u}$ is convex for $0<u\leq v$.
  
  (ii) We just consider $u^+$. Since~$S$ is a martingale, we have $u^+\geq1$.
  For any number $u\geq 1$ and $0<T\leq T'$, we have
  \[
    \mathbb{E}[S_{T'}^u]=\mathbb{E}\big[\mathbb{E}[S_{T'}^u|\mathcal{F}_T]\big]
    \geq \mathbb{E}[S_{T}^u]
  \]
  by the conditional Jensen inequality. This shows the assertion.
  
  (iii) We just prove the first statement. Note that (i) implies that~$\mathcal D^+$ is an interval.
  Now suppose for contradiction that $u\in\mathcal D^+$ satisfies
  \[
    u < u^+(T^*(u))
    =\sup\{ v: \mathbb{E}[S_{T^*(u)}^v]<\infty\}.
  \]
  This means that there is $v>u$ satisfying $\mathbb{E}[S_{T^*(u)}^v]<\infty$.
  Hence, $T^*(v)= T^*(u)$, contradicting the strict decrease of~$T^*$.
  Finally, suppose for contradiction that $u\geq 1$ satisfies
  $u > u^+(T^*(u))$. Then there is $1\leq v<u$ such that $\mathbb{E}[S_{T^*(u)}^v]=\infty$.
  For arbitrary $t\geq T^*(u)$, we get
  \[
    \mathbb{E}[S_t^v]=\mathbb{E}\big[\mathbb{E}[S_t^v|\mathcal{F}_{T^*(u)}]\big]
    \geq \mathbb{E}[S_{T^*(u)}^v]=\infty.
  \]
  This implies $T^*(v)=T^*(u)$, again contradicting the strict decrease of~$T^*$.
\end{proof}
If the assumptions of part~(iii) hold, then we can compute the critical moments
from the explosion time, by numerically solving the equations
$T^*(u^+(T)) = T$ resp.\ $T^*(u^-(T)) = T$.
Note, however, that \emph{strict} monotonicity may fail for
reasonable stochastic volatility models.
In the $3/2$-model~\cite{Le00}, the explicit characteristic function shows that
the critical moments do not depend on maturity (for positive maturity),
and the explosion time assumes only the values zero and infinity.

In the classical Heston model, on the other hand, it easily follows from~\eqref{eq:He et} that
$T_1^*(u)$ (recall that the index denotes $\alpha=1$)
is a strictly monotonic function of~$u$: On the set where it is finite,
$T_1^*$ strictly increases for negative~$u$, and strictly decreases for positive~$u$.
By part~(iii) of Lemma~\ref{le:gen}, this implies that we have
\begin{equation}\label{eq:mom eq}
  T_1^*(u^+(T)) = T \quad \text{and} \quad T_1^*(u^-(T)) = T, \quad T>0.
\end{equation}
Thus, although the critical moments do not admit an explicit expression,
they can be computed using~\eqref{eq:He et} and~\eqref{eq:mom eq},
 by numerical root finding with an appropriate starting value.
 
% We will now show that $T_\alpha^*(u^-(T)) = T$ in the rough Heston model, if $T$ is not too big, namely
% such that $u^-(T)$ satisfies case~(A). The latter restriction can certainly
% be removed with some extra work, but so far this is of little interest, because we have
% no algorithm to compute $T_\alpha^*$ in case~(B).

While we have no doubt that strict monotonicity of the explosion time extends
to the \emph{rough} Heston model, this seems not easy to verify. If we accept it as given,
then the lower critical moment can be computed for $\rho<0$ from
\begin{equation*}\label{eq:u- rough}
    T_\alpha^*(u^-(T)) = T.
  \end{equation*}
Again, we focus on the \emph{lower} critical moment, because then we can apply
Algorithm~\ref{algo:T} to compute $T_\alpha^*$ for $\rho<0$. Recall that this algorithm works
only for $u\leq \lambda/(\rho\xi)$, which amounts to case~(A).
Thus, $T$ must not be too large in~\eqref{eq:u- rough}, namely such that $u^-(T)$
satisfies case~(A). (Usually, this requirement is not too prohibitive.)

To provide some indication for the strict monotonicity of~$T_\alpha^*$,
recall that according to Lemma~\ref{le:mon}, $ f(u, t) = c_3 \psi(u, t) $
(see Section~\ref{se:prelim}) decreases strictly w.r.t.~$u$, if $u<0$ satisfies
case~(A). It is then plausible (although
not proven) that the strictly smaller function $ f(u_2,\cdot)$ explodes
at a larger time than $ f(u_1,\cdot)$, where $u_1<u_2<0$.
As another indication, the bounds in Theorems~\ref{thm:lb} and~\ref{thmIII:T*-upper}
are strictly monotonous, as seen by differentiating them w.r.t.~$u$.

Even in case strict monotonicity should not hold, we certainly have
\begin{align}
  u^-(T) &= \sup\{ u<0: T_\alpha^*(u)=T\} \label{eq:u sup} \\
  u^+(T) &= \inf\{ u>1: T_\alpha^*(u)=T\} \label{eq:u inf}
\end{align}
for all $T>0$, which suffices for numerical computations (under the above restriction
on~$T$). The validity of~\eqref{eq:u sup} and~\eqref{eq:u inf} is clear
from~\eqref{eq:blowup u}: If $T^*_\alpha(\cdot)$ is constant
on some interval, lying to the left of zero, say, then the mgf blows
up as~$u$ approaches the interval's right endpoint from the right.

\section{Application to asymptotics}\label{se:app}

In the introduction we mentioned several potential applications of our work.
In this section, we give some details on one of them:
Knowing the critical moments gives first order asymptotics for the implied volatility
for large and small strikes. We write $\hat{\sigma}(k)$ for the implied volatility,
where $k=\log(K/S_0)$ is the log-moneyness. According to Lee's moment formula~\cite{Le04a},
the left wing of implied volatility satisfies
\begin{equation}\label{eq:lee}
  T\cdot \limsup_{k\to-\infty} \frac{\hat{\sigma}(k)^2}{|k|}
    = 2-4\big(\sqrt{ u^-(T)^2- u^-(T)} + u^-(T)\big).
\end{equation}
We focus on negative log-moneyness, because then the slope depends on the \emph{lower}
critical moment, which Algorithm~\ref{algo:T} computes in the important case $\rho <0$.
%\begin{equation}\label{eq:lee}
%  T \limsup_{k\to\infty} \frac{\hat{\sigma}(k)^2}{k}
%    = 2-4\big(\sqrt{ u^+(T)^2- u^+(T)} - u^+(T)+1\big).
%\end{equation}
As in any model with finite critical moments, the marginal densities of the rough
Heston model have power-law tails.
More precisely, if we write $f_T$ for the density of~$S_T$, then
\begin{equation*}%\label{eq:dens infty}
  f_T(x)  = x^{-u^+(T)-1+o(1)},\quad x\to\infty,
\end{equation*}
and
\begin{equation}\label{eq:dens zero}
  f_T(x)  = x^{-u^-(T)-1+o(1)},\quad x \downarrow 0.
\end{equation}
Our approach (see Section~\ref{se:comp mom}) allows to evaluate the right-hand sides of~\eqref{eq:lee}
and~\eqref{eq:dens zero} numerically for the rough Heston model, if~$T$ is not too large.

In~\cite{FrGeGuSt11}, \eqref{eq:lee}--\eqref{eq:dens zero} were considerably sharpened
for the classical Heston model. We expect that such a refined smile expansion
can be done for rough Heston, too, with density asymptotics of the form
\[
  f_T(x) \sim c_1 x^{-u^+(T)-1} e^{c_2 (\log x)^{1-1/(2\alpha)}} (\log x)^{c_3},
  \quad x\to\infty,
\]
where the $c_i$ depend on~$T$ and~$\alpha$.
In the classical Heston model, the factor $e^{c_2 (\log x)^{1-1/(2\alpha)}}$ becomes
$e^{c_2 \sqrt{\log x}} ,$ in line with~\cite{FrGeGuSt11}.
Extending the analysis of~\cite{FrGeGuSt11} to $\tfrac12<\alpha<1$ will require
a detailed study of the blow-up behavior of the Volterra integral equation~\eqref{eq:vie}.
Among other things, (a special case of) the heuristic analysis in~\cite{RoOl96},
which we already mentioned in Section~\ref{se:comp exp}, would have to be made
rigorous, and extended to ensure uniformity w.r.t.\ the parameter~$u$. 
We postpone this to future work. Note that the approximation~\eqref{eq:Li} might be useful in this context.

%\section{The Volterra-Heston model}

%\textcolor{blue}{TODO (Arpad)}

\bibliographystyle{siam}
\bibliography{gerhold}
%\bibliography{../gerhold}
%\bibliography{./gerhold_cmod_2018-01-22}

\end{document}